\documentclass[a4paper,12pt, notitlepage,fleqn]{article}

\usepackage{graphicx}
\usepackage{verbatim}
\usepackage[longnamesfirst, sort]{natbib}
\usepackage{fancyhdr}
\usepackage[T1]{fontenc}
\usepackage{subcaption}
\usepackage[font={small,sc}]{caption}
\usepackage{amsmath}
\usepackage{amssymb}
\usepackage[amsmath,thmmarks]{ntheorem}
\usepackage{bbm}
\usepackage{mathrsfs}
\usepackage{enumitem}
\usepackage{tikz}
\usepackage{pgfplots}
\pgfplotsset{compat=1.18}
\usepgfplotslibrary{fillbetween}
\usetikzlibrary{pgfplots.groupplots}
\usetikzlibrary{external} 
\usetikzlibrary{patterns}
\usepackage{accents}
\usepackage{setspace}
\usepackage{multirow,array}
\usepackage{xcolor}
\usepackage[normalem]{ulem}

\usepackage{mdframed}

\newmdenv[linecolor=black,backgroundcolor=gray!10,linewidth=1pt]{mybox}

\newcommand{\E}{\mathbb{E}}

\newcommand{\derivative}[2]{\frac{\partial {#1}}{\partial {#2}}}

\newcommand{\bmat}[1]{\left[\begin{array}{#1}}
\newcommand{\emat}{\end{array}\right]}

\allowdisplaybreaks[1]

\newtheorem{corollary}{Corollary}

\newtheorem{definition}{Definition}
\newtheorem{example}{Example}

\newtheorem{lemma}{Lemma}

\newtheorem{proposition}{Proposition}

\newtheorem{assumption}{Assumption}
\newenvironment{proof}[1][Proof]{\noindent\textbf{#1.} }{\ \rule{0.5em}{0.5em}}

\begin{document}
\begin{titlepage}

    \title{A Theory of Likes}

    \date{\today}

    \author{Jean-Michel Benkert and Armin Schmutzler\thanks{Benkert: Department of Economics, University of Bern, jean-michel.benkert@unibe.ch. Schmutzler: Department of Economics, University of Zurich, and CEPR, armin.schmutzler@econ.uzh.ch. This paper was previously circulated under the title ``A Theory of Recommendations''. We are grateful for very helpful discussions with Dirk Bergemann, Piotr Dworczak, Johannes Johnen, Igor Letina, Simon Loertscher, Matt Mitchell, Marc M\"{o}ller, Nick Netzer, Volker Nocke, Martin Peitz, Marcus Reisinger, Egor Starkov, Jakub Steiner and seminar participants in Bamberg, Bern, Fribourg, Lausanne, Zurich, at MaCCI (2025), SAET (2023), Swiss Society of Economics of Statistics (2024), Swiss IO Day (2024) and Verein für Socialpolitik (2024).}}
    \vspace{1cm}
    \maketitle
    \begin{abstract}
       We study binary, non-strategic ratings when consumer preferences reflect a combination of objective quality differences between products and idiosyncratic taste. A rating then decomposes into objective and subjective content. Rating systems that are extremely strict or extremely lenient only convey objective content, maximizing consumer welfare when the population is not biased toward a particular product. With a biased population, less extreme rating systems that convey subjective content can be valuable. In many environments, a sales-maximizing platform prefers lenient rating systems, creating a conflict with welfare maximization.
    \end{abstract}

    


    \end{titlepage}
    
\newpage 
\onehalfspacing
    \section{Introduction}

Simple peer-to-peer ratings play an essential role for the diffusion of economic information in diverse domains, most notably in the digital economy. In some ways, this is surprising. For example, consider a consumer deciding whether to ``like'' a restaurant after a meal. Her experience may depend on at least two dimensions, such as food quality and service quality, each of which can be good or bad. Consumers agree that good food is better than bad food and good service is better than bad service, but they may differ in how much weight they place on each dimension. The consumer leaving the rating has little incentive to think hard about how her like will be interpreted, while the consumer receiving the rating may rely on it to avoid a poor match. Moreover, even a thoughtful consumer cannot convey all relevant nuances through a coarse binary signal. The example illustrates three core features of  peer-to-peer rating: \textit{Products and preferences are heterogeneous}, \textit{stakes are asymmetric} for those who leave ratings and those who read them, and \textit{ratings are inherently coarse}.

In this paper, we study how such ratings can nonetheless provide value. We first show how coarse rating systems can be designed in such a way that a like (dislike) conveys objectively good (bad) news about a product. We then argue that, from a consumer welfare perspective, it need not be optimal to make use of this option. The subjectivity of ratings is not always a bug to be eliminated, but can be the feature that makes a rating system valuable: If the population is known to be biased toward particular products, it may be desirable to design the system in such a way that consumers disagree about whether a like is good or bad news about a product. Consumers with majority taste will then understand that they should buy such a product, those with minority tastes understand that they should avoid it. Having identified the optimal design of the rating system from a consumer perspective, we then find that a sales-maximizing platform will often deviate from this prescription, relying too often on excessively lenient rating systems.

Our model captures the three core features of the example in a parsimonious way. An agent must decide between two \emph{products}, which can come in different variants. There are  objectively good (bad) variants, for which all consumers agree that they are better (worse) than all alternatives. Moreover, there are two controversial variants with intermediate values, which consumers rank differently. Ex ante, consumers merely have distributional information on which variant they are facing. However, one of the  products comes with a rating from a previous consumer: This \emph{sender} mechanically gives a like for a product that is objectively good or at least sufficiently good from her perspective (with payoffs at least some exogenous rating standard $R$). The \emph{receiver} of the signal knows the preference distribution and therefore how likely it is that his preferences are aligned with those of the sender. Finally, based on this information, he decides whether to be obedient, that is, behave according to the prescription of the rating.

A transparent cut-off condition determines whether a receiver is \textit{obedient} or \textit{disobedient}. Importantly, the condition is the same for likes and dislikes. The key insight is that ratings contain both objective and subjective content. From a (dis-)like, a receiver learns that the product is not objectively bad (good), generating a payoff effect that pushes every receiver toward obedience. A like usually also raises the probability that, if the product is controversial, it is of the variant preferred by the majority. For consumers with majority tastes, these two pieces of information are both good news, so that they are obedient. Consumers with minority tastes face opposing effects and are obedient only when the objective effect outweighs the subjective one, and disobedient otherwise. 

Our main contributions concern the design of rating systems. Platforms often try to influence how critical consumer ratings are, for example through guidelines or instructions.\footnote{For example, Uber employs a five-star system but informs riders that ``5 stars means the ride went without issues,'' while ``4 stars or less'' prompt requests for feedback \citep{Uber2024}. Similarly, Airbnb communicates that a 5-star rating means ``everything was as expected'' \citep{Hospitable2024}.} We capture this by assuming that the platform controls the rating standard $R$ above which a like is given. Our first contribution is that we characterize the value of rating systems for consumers and identify the value-maximizing choice of $R$. When the population is not biased toward any of the controversial products on average, the subjective effect of ratings cancels out in the aggregate, and extreme rating standards are optimal: Likes should be reserved for objectively good products, or dislikes for bad ones. When the population is biased, ratings that convey subjective content due to an interior rating standard may be strictly preferable. Intuitively, ratings can provide value in two ways. First, they can push consumers toward objectively good products and help them avoid objectively bad ones. Second, they can help consumers obtain a product variant that fits their tastes, even when it is not perfect. These two objectives need not be conflicting: Consumers with majority tastes will benefit in both ways from being obedient. Instead, for consumers with minority tastes, being obedient reduces the chance of getting the preferred controversial product variant. Being disobedient reduces the chance of getting the objectively good rather than the objectively bad variant. Either way, the rating system creates value for these consumers, even if they do the opposite of what they are told.

Our second main contribution concerns the potential conflict between consumer welfare and sales maximization by a platform that benefits from transactions of the rated product, for instance through transaction fees or advertising revenue tied to traffic. We find that a sales-maximizing platform usually favors a low $R$, issuing likes for everything but objectively bad products. By contrast, the value-maximizing $R$ can be high or interior. 

We provide several further results. First, we introduce prices and show that the firm's pricing decision affects not only demand but also the informativeness of the rating. Second, we analyze the case where receivers observe multiple independent ratings. Perhaps surprisingly, access to more signals does not necessarily increase value or sales. Third, we allow the sender and receiver populations to differ, showing that our main results remain robust while new design considerations arise regarding the composition of the sender pool. We close by examining how polarization and the prevalence of controversial products shape the value of the rating system.

\subsection{Related Literature}\label{sec:related_literature}

Our paper relates most directly to work distinguishing vertical quality from horizontal fit in online reviews. \cite{sun2011variance} argues that variance of ratings does not reflect pure noise but conveys information about taste differences. \cite{chen2021reviews} theoretically and empirically separate uncertainty about quality from uncertainty about fit, finding that online ratings resolve the former more effectively than the latter. Closest in spirit is \cite{lafky2024ratings}, who show experimentally that raters evaluate products through their own preferences and consumers interpret ratings accordingly. \cite{bohren2025beyond} provide complementary experimental evidence, showing that ratings raise welfare under vertical  but not under horizontal differentiation, while segmenting ratings by consumer type restores the welfare gains.\footnote{\cite{bohren2025beyond} isolate the pure vertical and pure horizontal cases, while we study environments in which both dimensions coexist within a single product space.}

In a broad sense, our static inference problem belongs to the literature on observational learning.\footnote{ However, this literature, e.g. \citet{banerjee1992simple} \citet{bikhchandani1992theory} and \citet{smith2000pathological}, focuses on dynamics, and it emphasizes learning from the decisions of others rather than ratings.} A more closely related literature deals specifically with learning from online reviews. For instance, \cite{acemoglu2022learning} show that consumers with strong prior preferences are more likely to purchase and review, so the pool of recommenders is endogenously selected and receivers must account for this bias. \cite{che2018recommender} study the tension between informativeness and exploration in recommender systems. We study a single binary signal in isolation, which makes the rating standard the sole design instrument, enabling us to characterize precisely how it trades off objective against subjective content.

Our design problem is naturally framed as one of constrained information design: The platform cannot choose an arbitrary signal structure \citep{kamenica2011bayesian}, only the rating standard at which senders convert their experience into a like. This can be justified using results on the optimality of coarse disclosure in certification and rating design, including partial information revelation \citep{rayo2010optimal}, monopolistic pass/fail certification \citep{lizzeri1999information}, and optimal partition of quality into rating tiers \citep{hopenhayn2026optimal}. Unlike this prior work, our signal is generated mechanically from the experiences of heterogeneous senders, and the central trade-off is not more versus less information but objective versus subjective content.\footnote{Strategic senders can produce coarse messages endogenously, see, e.g., \cite{chakraborty2010cheap} and \cite{smirnov2024designing}. Our senders are non-strategic, so coarseness is imposed by the binary rating rule rather than chosen.}

Our analysis of the tradeoff between value and sales maximization connects to the literature on platform economics and rating bias, where platform incentives often diverge from consumer welfare. The most closely related papers in this area also work with heterogeneous tastes. \cite{shi2020economic} study \textit{which} product a platform should feature, whereas we study \textit{how demanding} the standard for a positive recommendation should be. \cite{PeitzSobolev2024} show that a profit-maximizing firm may optimally recommend products even to poorly matched consumers.  A literature on reputation and recommendation bias \citep{tadelis2016reputation, gaudin2022streaming} documents how feedback systems can be distorted by strategic behavior.\footnote{For instance, \cite{bolton2013engineering} show how the threat of retaliation in reciprocal feedback systems leads buyers to withhold negative feedback.} By contrast, the wedge between platforms and consumers that we identify does not rest on fake reviews or inflated individual messages. Ratings are mechanically truthful, and the distortion arises because a sales-maximizing platform chooses a more permissive rating standard than a value-maximizing one.

Section~\ref{sec:model} introduces the model. Section~\ref{sec:preliminary} characterizes receiver behavior. Section~\ref{sec:value_design} studies the value of rating systems, first characterizing the sources of value and then analyzing the optimal design of the rating standard. Section~\ref{sec:sales} turns to the parallel question of sales maximization, where a platform seeks to maximize transaction volume rather than consumer welfare. Section~\ref{sec:extensions} presents further results. Section~\ref{sec:conclusion} concludes. All proofs are relegated to the appendix.

\section{Model}\label{sec:model}

\paragraph{Products and Payoffs}
   
We consider a consumer's choice between two products. Each product can be of one of four variants: $g,b,c_1$ and $c_2$. The probability of each variant is independently drawn from the same distribution $\mathbf{q} = (q_g, q_1, q_2, q_b)$ where
    \begin{equation*}
        \Pr(g)= q_g>0,  \Pr(b)=q_b>0, \Pr(c_1) = q_1>0, \Pr(c_2)= q_2>0.    
    \end{equation*}
Consumers cannot observe the product variant before consumption, but they know the distribution. Consumer types are distributed on $[-1/2, 1/2]$ according to $F$, which is absolutely continuous with full support. All consumers obtain a payoff of $1$ from consuming variant $g$ and a payoff of $0$ from consuming variant $b$. The payoff from consuming variant $c_1$ is $1/2+i$, the payoff from variant $c_2$ is $1/2-i$. Therefore, consumers agree that product variants $g$ and $b$ are (objectively) good and bad, leading to the maximal (minimal) payoff. By contrast, the payoff from $c_1$ and $c_2$ depends on type $i$, so that the assessment of these variants is subjective. We thus refer to them as \textit{controversial}. Types $i>0$ prefer $c_1$ to $c_2$, and vice versa for $i<0$, while type $i=0$ is indifferent between the two variants.

We use the following self-explanatory terminology:
    \begin{definition}
        The \textbf{prevalence of controversial products} is given as $Q:=\frac{q_1+q_2}{2}$. The \textbf{odds of a good product} are $\sigma:= q_g/q_b$.
    \end{definition}

\paragraph{A Motivating Example}
To motivate the payoff structure, one might assume that a product variant has two quality dimensions $d \in \{1,2\}$ that are relevant to consumers. The quality level in each dimension $d$ can take values $\omega_d \in \{0,1\}$. All consumers agree that high quality in any dimension is desirable, but they put weight $\frac{1}{2}+i$ on dimension $1$ and weight $\frac{1}{2}-i$ on dimension $2$. Defining $c_i$ as a product variant with quality $1$ in dimension $i$ and quality $0$ in dimension $j \ne i$, the above payoff structure obtains.
    
\paragraph{Ratings and Updating}
Before choosing products, a consumer receives a rating for one of the two products from a previous consumer. This \textit{sender} (she) and the \textit{receiver} (he) are both drawn from the distribution $F$. Senders and receivers who consume the same product face the same variant and differ only in their type-dependent evaluation of controversial products. The receiver chooses between the product with a rating $r$, which can be a like ($r=L$) or a dislike ($r=D$), and the alternative without ($r=0$). 

To eliminate any scope for strategic sender behavior, we assume that a sender mechanically gives a like if she obtained a payoff of at least $R\in (0,1)$ and a dislike otherwise.\footnote{We exclude the rating standard $R=0$ as this would render the rating uninformative. For symmetry, we also exclude $R=1$, but all results would essentially go through if we allowed $R=1$.} Intuitively, a like conveys the following information to the receiver: The probability of the bad product variant falls to zero, the likelihood of the good variant relative to the controversial variants rises, and the odds between the controversial variants change.

For the objectively good and bad variants $g$ and $b$, there is \textit{sender agreement}: All types like the former and dislike the latter. For the controversial variants, sender behavior is type-dependent: After consuming $c_1$, a sender of type $i$ gives a like if and only if $i \geq R-1/2$, whereas types $i \leq 1/2-R$ give a like after consuming $c_2$; see Figure~\ref{fig:senders}. A rating has \textit{objective content} on which all senders agree, and \textit{subjective content} on which they do not. The rating standard determines the mix between objective and subjective content: For interior $R$, sender disagreement on controversial products is preserved and ratings carry nontrivial subjective content. As $R\to 0$ or $R\to 1$, sender disagreement vanishes and ratings become purely objective in content.

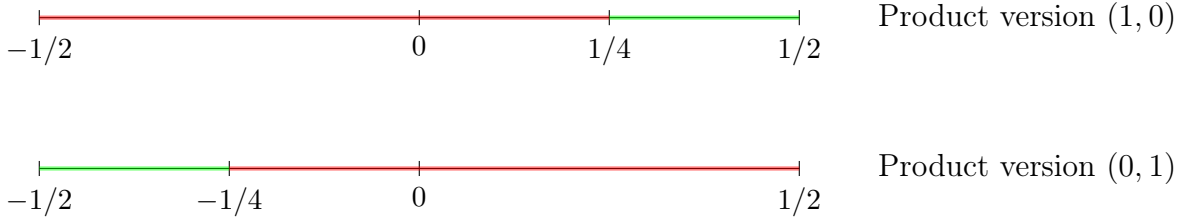
\begin{figure}[htb]
\centering
\caption{Illustration of Sender Behavior}
\label{fig:senders}
    \begin{tikzpicture}[scale=10]
    \draw (-0.5,0.2) -- (0.5,0.2);
    \foreach \x in {-1/2, 0, 1/4, 1/2}
        \draw (\x,0.19) -- (\x,0.21) node[below=0.2cm, font=\small] {$\x$};
        
    \node at (0.8,0.2) {Product variant $c_1$};
    \draw[ultra thick, green, opacity=0.5] (1/4,0.2) -- (1/2,0.2);
    \draw[ultra thick, red, opacity=0.5] (-0.5,0.2) -- (1/4,0.2);

    \draw (-0.5,0) -- (0.5,0);
    \foreach \x in {-1/2, -1/4, 0, 1/2}
        \draw (\x,-0.01) -- (\x,0.01) node[below=0.2cm, font=\small] {$\x$};
        
    \node at (0.8,0) {Product variant $c_2$};
    \draw[ultra thick, green, opacity=0.5] (-1/2,0) -- (-1/4,0);
    \draw[ultra thick, red, opacity=0.5] (-1/4,0) -- (0.5,0);
\end{tikzpicture}

\flushleft
\footnotesize \textit{Notes:} The figure illustrates what types send a like (green) and a dislike (red) for the two controversial product variants when $R=3/4$.
\end{figure}
       
Let $\phi_1(R):= 1-F(R-1/2)$ and $\phi_2(R):= F(1/2-R)$ denote the probability of a like conditional on the product variant being $c_1$ or $c_2$, respectively. The probability that a randomly chosen sender gives a like or dislike, respectively, is 
    \begin{align*}
       \pi^L(R) &:=  q_g+q_1\phi_1(R)+q_2\phi_2(R), \\
       \pi^D(R) &:= 1-\pi^L(R)= q_1(1-\phi_1(R))+q_2(1-\phi_2(R))+ q_b.
    \end{align*}
 In words, the probability of receiving a like from a randomly chosen sender is given by the sum of the probability that a product is objectively good and the probability that it is subjectively good enough from a randomly chosen sender's perspective. We assume that receivers are Bayesian. In Appendix \ref{Sec:App_Posteriors}, we present the receiver's posterior $\mathbf{p}^r(R):= (p_g^r(R), p_1^r(R),p_2^r(R),p_b^r(R))$ after a like $r=L$ or a dislike $r=D$.\footnote{To simplify notation, we will often drop the dependence of these posteriors and other variables on $R$.}

 Finally, we use the following summary terminology:
\begin{definition} \label{def_RS}
    $\mathcal{E}:=(\textbf{q},F)$ defines a \textbf{decision environment}. A \textbf{rating system} $\mathcal{R}$ consists of a decision environment $\mathcal{E}$ and a rating standard $R \in (0,1)$.
\end{definition}
Throughout, the decision environment $\mathcal{E}$ and the rating standard $R$ are common knowledge among all agents.

\section{Receiver Behavior} 
\label{sec:preliminary}\label{sec:receiver_behavior}

The receiver must decide whether to follow the prescription of the rating, that is, whether to purchase the rated product after a like and avoid it after a dislike. We now characterize the receiver's optimal choice, depending on his type $i$ and the rating system $\mathcal{R}$. 

The expected payoffs from buying the product with a rating and the alternative without are respectively given as
\begin{align*}
        U^r(i)&:=p_g^r+(1/2+i)p_1^r  + (1/2-i)p_2^r, \quad r\in \{L,D\};\\
   U^0(i) &:=q_g+ (1/2+i)q_1 +(1/2-i)q_2 .
    \end{align*}
Receiver $i$ buys the product rated with a like if and only if $U^L(i) \geq U^0(i)$ or equivalently
 \begin{align}
    \frac{p_g^L-q_g}{2} -\frac{p_b^L-q_b}{2} + i[(p_1^L-q_1) -(p_2^L-q_2) ] \geq  0\label{eq:optimal_following}. 
    \end{align}

The receiver follows the prescription of dislike ratings if and only if $U^0(i) \geq U^D(i)$. A rating's informational impact on the receiver decomposes into an objective and a subjective payoff effect. To make this precise:
     \begin{definition}\label{def:effects} For a rating $r \in \{D,L\}$, let $\Delta_O^r$ and $-i\Delta_S^r$ denote its objective and subjective payoff effects, respectively, where 
          \begin{align*}\label{eq:objective}
       \Delta_O^r&:=\frac{p_g^r-q_g}{2} -\frac{p_b^r-q_b}{2},\\
        \Delta_S^r&:=(p_2^r-q_2) -(p_1^r-q_1)  .
     \end{align*}
    \end{definition}
If $\Delta_O^r > 0$, the rating $r$ has shifted probability from the objectively bad to the objectively good variant, a payoff gain shared by all types. $\Delta_S^r$ measures the analogous shift between the controversial variants: from $c_1$ to $c_2$ when positive, and conversely when negative. Whether this shift is good news depends on the receiver: The sign of the subjective payoff effect, $-i\Delta_S^r$, depends on the signs of $i$ and $\Delta_S^r$, so that a shift toward $c_2$ is welcome only for types who prefer $c_2$ to $c_1$. Applying this notation to condition~\eqref{eq:optimal_following}, we obtain:

 \begin{proposition}\label{prop:equivalence_following} For a given rating system $\mathcal{R}$, a receiver of type $i$ buys the rated product after a like and does not buy it after a dislike if and only if
\begin{equation}\label{eq:optimal_following_short}
  \Delta_O^L \geq i\Delta_S^L.   
\end{equation}
    \end{proposition}
    
Intuitively, a receiver follows a prescription only if he expects the rating to come from a sender whose preferences are sufficiently aligned with his. Condition \eqref{eq:optimal_following_short} formalizes the resulting trade-off via the objective and subjective payoff effects. Importantly, this condition is the same for likes and dislikes, which motivates the following terminology:

\begin{definition}$ $
\begin{enumerate}
   \item[(i)] A receiver is \textbf{obedient} if he buys the rated product after a like but not after a dislike, \textbf{disobedient} if he buys it after a dislike but not after a like.
    \item[(ii)] Receivers are \textbf{universally obedient} in decision environment $\mathcal{E}$ if all receivers are obedient for every $R\in(0,1)$.
\end{enumerate}
    \end{definition}
Denoting the type who is indifferent between being obedient or disobedient as
\begin{align*}
         \tilde{i}:=& \frac{\Delta_O^L}{\Delta_S^L} \text{ if } \Delta_S^L \neq 0\text{,}
    \end{align*}
Proposition \ref{prop:equivalence_following} immediately implies the following result:
\begin{corollary}\label{cor:follow} Fix a rating system $\mathcal{R}$.
   \begin{enumerate}
            \item[(i)] If $\vert\Delta_S^L \vert\leq 2\Delta_O^L$, then all $i \in [-1/2,1/2]$ are obedient.
            \item[(ii)] If $\Delta_S^L< -2\Delta_O^L$, then all $i \geq  \tilde{i}$ are obedient.
            \item[(iii)]  If $\Delta_S^L> 2\Delta_O^L$, then all $i \leq  \tilde{i}$ are obedient. 
        \end{enumerate}  
\end{corollary}
Intuitively, in case (i), the objective effect is so large that it dominates even the subjective effect of the most extreme receivers. As $\Delta_S^L$ and hence the subjective effect increases, receivers who are not sufficiently aligned with the population at large will be disobedient; see cases (ii) and (iii).
The interplay between objective and subjective effect leads to a stark result for extreme rating standards:
\begin{proposition} \label{prop:extreme_recommendations}
        Suppose $R\rightarrow 1$ or $R\rightarrow 0$. Then, all types are obedient irrespective of the decision environment $\mathcal{E}$. 
    \end{proposition}    
In the limits, ratings have purely objective content. For instance, as $R \to 1$, a like guarantees the product is objectively good, while a dislike implies it is not. 

\section{Value of Rating Systems}
\label{sec:value_design}

This section takes a welfare perspective, studying the value of rating systems for receivers and deriving the value-maximizing choice of the rating standard.

\subsection{Characterization of the Value}

\label{sec:value_RS}
Without a rating system, all product variants are identical ex ante, with an expected payoff based on the prior distribution of product variants:
    \begin{align}\label{eq_VNR}
      V^0(i):=U^0(i).
    \end{align}
If one of the products comes with a rating, an obedient receiver $i$ chooses this product (rather than the alternative) only if it comes with a like. The expected payoff is thus
    \begin{align}\label{eq_VF}
        V^+(i) := \pi^LU^L(i)+ (1-\pi^L)U^0(i).
    \end{align}
The first term on the right-hand side comes from likes, the second one from dislikes.

A disobedient receiver will buy a product that comes with a dislike. Otherwise, he will buy the alternative. Hence, his expected payoff is
    \begin{align}\label{eq_VNF}
       V^-(i) :=&\pi^LU^0(i)+ (1-\pi^L)U^D(i).  
    \end{align}

A like does not create value for disobedient receivers, who choose the alternative without a rating. Instead, a dislike, which is given with probability $(1-\pi^L)$, influences their behavior: Rather than buying an alternative based on priors, updating leads them to choose the product for which they received the explicit recommendation not to buy. 

We now define the value of a rating system.\footnote{This definition is in line with the definition of a valuable signal in \citet{kamenica2011bayesian}.}
\begin{definition} \label{def:value_RS}
    $V(\mathcal{R})$, the \textbf{value of a rating system }$\mathcal{R}= (\mathcal{E},R)$, is the increase in expected payoffs resulting from the existence of a rating standard $R$ in a decision environment $\mathcal{E}$, where the expectation is taken over all pairs of independently drawn senders and receivers.
\end{definition}
Using Definition \ref{def:value_RS} and Corollary~\ref{cor:follow}, we can characterize $V(\mathcal{R})$ as follows\vspace{2mm}:

    \begin{tabular}{lll}
    (i) & $\int_{-1/2}^{1/2} (V^+(i) - V^0(i))dF(i)$  & if  $ \vert\Delta_S^L \vert\leq 2\Delta_O^L$; \\
    (ii) & $\int_{-1/2}^{\tilde{i}}V^-(i) dF(i)  +\int_{\tilde{i}}^{1/2} V^+(i) dF(i) - \int_{-1/2}^{1/2} V^0(i)dF(i)$ & if  $\Delta_S^L<-2\Delta_O^L$; \\
    (iii) & $\int_{-1/2}^{\tilde{i}}V^+(i) dF(i) +  \int_{\tilde{i}}^{1/2} V^-(i) dF(i)  - \int_{-1/2}^{1/2} V^0(i)dF(i)$ & if  $\Delta_S^L>2\Delta_O^L$.
\end{tabular}\\

\noindent All receivers are obedient when $ \vert\Delta_S^L \vert\leq 2\Delta_O^L$, as the objective effect dominates (case (i)). Hence, by equation~\eqref{eq_VF}, the rating system creates value only for a receiver who observes a like, as a dislike leads to the same expected payoff as without a rating system. Parts (ii) and (iii) refer to cases where the subjective effect of the messages matters more. Then, not all receivers are obedient and, as argued above, the system can create value for receivers that are obedient as well as those who are not.

We restate the value $V(\mathcal{R})$ to better understand how it reflects the role of the objective and subjective payoff effects of ratings:\footnote{As we focus on $R$ as the design parameter, we often write $V(R)$ instead of $V(\mathcal{R})$.}

\begin{proposition} \label{prop:value_of_RS}
       The value $V(R)$ of the rating system $\mathcal{R}$ is:\\
\begin{small}
\medskip
    \begin{tabular}{lll}
    (i) & $\pi^L[\Delta_O^L-\Delta_S^L \E[i]]$  & if  $ \vert\Delta_S^L \vert\leq 2\Delta_O^L$; \\
    (ii) & $(1-\pi^L)F(\tilde{i}) \left[\Delta_O^D  - \Delta_S^D\E[i \mid i\leq \tilde{i}]\right] + \pi^L(1- F(\tilde{i}))\left[\Delta_O^L - \Delta_S^L  \E[i \mid i\geq \tilde{i}]\right]$ & if  $\Delta_S^L<-2\Delta_O^L$; \\
    (iii) & $\pi^LF(\tilde{i})\left[\Delta_O^L  - \Delta_S^L \E[i \mid i\leq \tilde{i}]\right]+ (1-\pi^L)(1- F(\tilde{i})) \left[\Delta_O^D - \Delta_S^D \E[i \mid i\geq \tilde{i}]\right]$ & if  $\Delta_S^L>2\Delta_O^L$.
\end{tabular}
\end{small}
    \end{proposition}
The positive payoff effect of the rating system in case (i) consists of the type-independent objective effect $\Delta_O^L$ and the expected subjective effect $-\Delta_S^L\mathbb{E}[i]$. As $\Delta_S^L$ increases in absolute terms, some receivers become disobedient. For instance, when $\Delta_S^L > 2\Delta_O^L$, a like shifts beliefs toward the controversial product variant that is less preferred by types $i>0$, who may thus become disobedient. Only the  fraction $F(\tilde{i})$ of the population whose preferences are sufficiently aligned with the sender’s ($i<\tilde{i}$) still buy. For them, the subjective effect is either positive or not sufficiently negative to dominate the positive objective effect. The average subjective payoff effect among the obedient types is $-\Delta^L_S\E[i \mid i \leq \tilde{i}]$. As argued above, the fraction $1-F(\tilde{i})$ of disobedient types may benefit from the rating. They infer from a dislike that the product is likely to be of the controversial variant they personally prefer and thus buy despite or rather because of the negative signal.\footnote{This is reflected in the term $(1-\pi^L)(1- F(\tilde{i})) \left[\Delta_O^D - \Delta_S^D \E[i \mid i\geq \tilde{i}]\right]$ in Proposition \ref{prop:value_of_RS}(iii).}

\subsection{Design} \label{sec:value_design_sub}

The design problem is to choose a rating standard $R\in(0,1)$.\footnote{Throughout, we formulate the design problem on the open interval $R\in(0,1)$. Statements favoring extreme standards should therefore be understood as boundary statements, i.e.\ as limits or maxima of the continuous extension at $R=0$ or $R=1$.} Extreme standards yield sender agreement and thus purely objective content, while interior standards preserve disagreement and admit subjective content. We organize the analysis along this dichotomy. We first show that, when the population is symmetric, the aggregate contribution of the subjective effect cancels out, so design is governed by the objective effect alone and extreme standards are generically optimal. We next show that, when the population is asymmetric, the subjective effect no longer cancels out in the aggregate, and balancing it against the objective effect yields an interior optimum under simple general conditions on the decision environment. Finally, we present a tractable parametric example in which both regimes appear and the optimal rating standard can be characterized in closed form.

\subsubsection{Symmetric populations}\label{sec:Sym}

We begin with symmetric populations, so that preference heterogeneity corresponds to pure noise on average.
\begin{assumption} \label{ass:symmetry}
    $F(-i) = 1-F(i)$ \textit{for all} $i\in [-1/2, 1/2]$.
\end{assumption}
\noindent Under symmetry, the share $1-F(R-1/2)$ of senders who give a like after consuming $c_1$ is the same as the share  $F(1/2-R)$ who give a like after consuming $c_2$, and we define $\beta:=F(1/2-R)=1-F(R-1/2)$. Symmetry implies that the subjective effect of ratings cancels out in the aggregate. From Proposition \ref{prop:value_of_RS}, we thus obtain:
\begin{lemma}\label{prop:pop_symmetry_allfollow}
        For any rating system satisfying Assumption \ref{ass:symmetry}, all types are obedient and $ V(R) =\pi^L\Delta_O^L$.
    \end{lemma}
The optimal rating standard then depends only on $\sigma$, the odds of a good product:
  \begin{proposition} \label{prop:solution_symmetric_population}
        For any rating system satisfying Assumption \ref{ass:symmetry}, $V(R)$ is
        \begin{enumerate}
            \item[(i)] decreasing in $R$ if $\sigma<1$;
            \item[(ii)] increasing in $R$ if $\sigma>1$;
            \item[(iii)] independent of $R$ if $\sigma=1$.
        \end{enumerate}
    \end{proposition}
Intuitively, since the subjective effect of ratings cancels out in the aggregate, the designer can focus on the objective effect. To this end, he either chooses a very low rating standard close to $0$ that reveals whether a product is objectively bad or not, or a very high rating standard close to $1$ that reveals whether a product is objectively good or not. Which choice is better depends on how important these two pieces of information are: When objectively bad goods are common ($\sigma<1$), it is important to avoid them, so an extremely low rating standard is appropriate; when objectively good products are common ($\sigma>1$), it is important to identify them through extremely high rating standards.

\subsubsection{Asymmetric populations}\label{sec:Asym}

We now turn to asymmetric populations. To capture biased populations in a tractable way, we introduce the following weak assumption.

\begin{assumption}
\label{ass:full_support_bounded_density}
$F$ has mean $\mu:=\E[i]\neq 0$.
Moreover, its density $f$ extends continuously to $[-\tfrac12,\tfrac12]$ and satisfies
$0<f(\pm\tfrac12)<\infty$.
\end{assumption}
We use a shorthand notation for the net payoff of the average receiver from consuming each controversial variant rather than taking the outside option:
\begin{align*}
    A_1&:= \left(\tfrac12+\mu\right)-U^0(\mu),\\
    A_2&:= \left(\tfrac12-\mu\right)-U^0(\mu).
\end{align*}
Here, $\tfrac12+\mu$ and $\tfrac12-\mu$ are the gross payoffs of type $\mu$ from variants $c_1$ and $c_2$, respectively, while $U^0(\mu)$ is her expected payoff from the unrated alternative. We obtain a sufficient condition for an interior rating standard to be optimal:

\begin{proposition}
\label{prop:asymmetric_interior_unequal_q}
For any rating system satisfying Assumption~\ref{ass:full_support_bounded_density}, if
\begin{equation}\label{eq:asymmetric_endpoint_condition}
   q_1A_1 f(-\tfrac12)+q_2A_2 f(\tfrac12)<0<q_1A_1 f(\tfrac12)+q_2A_2 f(-\tfrac12)
\end{equation}
then there exists an optimal $R^\ast\in(0,1)$.
\end{proposition}

Condition~\eqref{eq:asymmetric_endpoint_condition} is clearly inconsistent with Assumption~\ref{ass:symmetry}. Moreover, it requires that $A_1$ and $A_2$ have different signs. W.l.o.g., we assume that  $A_1>0>A_2$, meaning that variant $c_1$ is valuable for the average receiver relative to the outside option, whereas variant $c_2$ is not. Under this interpretation, condition \eqref{eq:asymmetric_endpoint_condition}
has a transparent meaning. The first inequality guarantees that rating standards near $R=0$ are dominated by slightly more stringent standards. Starting from an extremely lenient standard, a small increase in $R$ removes only likes from senders who have strong preferences for one of the controversial products and thus hardly value the other one. For variant $c_1$, these likes come from types near $-1/2$; for variant $c_2$, they come from types near $1/2$. The inequality states that these weakest marginal likes have negative average value: They are too heavily tilted toward recommendations for the controversial variant that is bad from a majority perspective.\footnote{The fact that senders near $1/2$  switch from like to dislike after observing variant $c_2$ is desirable from the perspective of the average receiver, the corresponding switch for senders near $-1/2$ after observing variant $c_1$ is not. The left inequality in \eqref{eq:asymmetric_endpoint_condition} guarantees that the beneficial effect of eliminating undesirable likes dominates the adverse effect of eliminating desirable likes. }  Hence, an extremely lenient standard is too permissive. By similar arguments, the second inequality guarantees that extremely stringent rating standards are not desirable.

\subsubsection{A tractable example}\label{sec:Bench}

The general results of Sections~\ref{sec:Sym} and~\ref{sec:Asym} establish the qualitative dichotomy between symmetric and asymmetric populations.
To illustrate the general mechanism from above, we now provide an explicit closed-form solution for a tractable parametric specification.\footnote{Example~\ref{ass:asymmetry} is complementary to Assumption~\ref{ass:full_support_bounded_density} rather than nested in it: the power distributions \(F(i)=(i+1/2)^a\) have boundary densities that vanish or diverge when \(a\neq1\). A tractable example of a distribution that satisfies Assumption~\ref{ass:full_support_bounded_density} is given by tilted-uniform distributions.}

\begin{example}\label{ass:asymmetry}
   (i) $F(i) = (i+1/2)^a$ for $a>0$; (ii)  $q_1=q_2= Q$.
 \end{example}
The case $a=1$ reflects a symmetric, unbiased population. By contrast, when $a \neq 1$, the population is biased toward one product variant or the other.\footnote{The condition $q_1=q_2$ ensures that all subjective content in ratings is driven by the asymmetry of the type distribution $F$, not by prior differences in how common the controversial variants are.}

Our characterization focuses on decision environments where the prevalence of controversial products is so low that receivers are universally obedient. In the Online Appendix \ref{App:universal_acceptance_cutoff}, we show that this holds if $Q\leq\bar Q(a,\sigma)$ for a  cutoff $\bar Q(a,\sigma)\in(0,1/2]$ that is unique given $a$ and $\sigma$, the other features of the decision environment. The condition is frequently satisfied.\footnote{If $a=1$, $\bar Q(1,\sigma)=1/2$, so universal obedience arises for every admissible $Q$. For $a\neq 1$, the cutoff is still quite large (for example, $\bar Q(2,1)=1/3$). At $\sigma=1$ and $a>1$ the closed form $\bar Q(a,1)=1/[4(1-2^{-a})]$ obtains (Corollary \ref{Cor_Obedience_Ex} in the Online Appendix). For general $\sigma$, $\bar Q(a,\sigma)$ is also computable from the primitives.}
Moreover, in Lemma \ref{lemma:cutoff_monotonicity_a} of the Online Appendix, we show that this condition is easiest to satisfy in the unbiased benchmark $a=1$ and becomes harder to satisfy as the population becomes more biased. We now describe the relation between the rating standard $R$ and the value, focusing on the asymmetric case.

  \begin{proposition} \label{prop:benchmark_characterization}
  Suppose the decision environment $\mathcal{E}$ satisfies the conditions in Example~\ref{ass:asymmetry} as well as  $Q\leq\bar Q(a,\sigma)$ and $a\neq 1$.
  \begin{enumerate}
      \item[(i)] If $Q\geq \min\left\{\frac{a}{a+1},\frac{1}{a+1}\right\}$, then $V(R)$ is maximized at a unique
  $R^\ast\in(0,1)$.
      \item[(ii)] If $Q<\min\left\{\frac{a}{a+1},\frac{1}{a+1}\right\}$, then
      \begin{enumerate}
          \item[(a)] $V(R)$ is decreasing in $R$ if and only if
  $\sigma\leq\min\left\{\frac{a-Q(a+1)}{1-Q(a+1)},\frac{1-Q(a+1)}{a-Q(a+1)}\right\}$;
          \item[(b)] $V(R)$ is increasing in $R$ if and only if
  $\sigma\geq\max\left\{\frac{a-Q(a+1)}{1-Q(a+1)},\frac{1-Q(a+1)}{a-Q(a+1)}\right\}$;
          \item[(c)] $V(R)$ is maximized at a unique $R^\ast\in(0,1)$ otherwise.
      \end{enumerate}
  \end{enumerate}
  \end{proposition}

Proposition~\ref{prop:benchmark_characterization} first identifies $\min\left\{\frac{a}{a+1},\frac{1}{a+1}\right\}$ as a cutoff on the share of controversial products, above which an interior optimum exists. Below that cutoff, extreme rating standards are optimal only when the odds of a good product are sufficiently unbalanced.  Figure~\ref{fig:interior_solution} illustrates how the occurrence of interior solutions depends on the parameters: the further $a$ is away from $1$, the larger is the interval of $\sigma$ values for which an interior solution exists. 

 In some cases, the restriction \(Q\leq \bar Q(a,\sigma)\) is not needed. In Example~\ref{ass:asymmetry}, if \(\sigma=1\), so that \(q_g=q_b\), the design role of the objective channel is shut down. The rating standard matters only through the subjective content of the rating. One can show that this content is maximized at \(R=1/2\), where senders like precisely the controversial variant they personally prefer. Thus the midpoint standard makes ratings maximally informative about the controversial dimension.\footnote{This remains true even when \(Q>\bar Q(a,1)\) and some receivers are disobedient, as disobedience does not affect the amount of subjective content conveyed by the rating.}

\begin{figure}[htb]
    \centering
        \caption{Illustration of Proposition~ \ref{prop:benchmark_characterization} }
    \label{fig:interior_solution}
\begin{tikzpicture}
        \begin{axis}[
            width=12cm,
            xmin=0, xmax=9,
            ymin=0, ymax=9,
            xlabel={$a$},
            ylabel={$\sigma$},
            axis lines=middle,
            axis line style={->},
            xlabel style={
                at={(axis description cs:1,0)}, anchor=west, font=\footnotesize
            },
            ylabel style={
                at={(axis description cs:0,1)}, anchor= south , font=\footnotesize
            },
            unit vector ratio=2 1, 
            xtick={0,1, 3, ..., 9}, 
            ytick={0, 1, 3, ..., 9},
            tick label style={font=\footnotesize}
        ]
        \addplot[ thick, domain=0.06:9, samples=200, dotted,name path=upper]{(0.95*x - 0.05)/(0.95 - 0.05*x)};
        \addplot[ thick, domain=0.06:9, samples=200, dashed,name path=lower]{(0.95 - 0.05*x)/(0.95*x - 0.05)};
        \node[font=\tiny, rotate=-55] at (axis cs:0.4,1.8) {$R^\ast \in (0,1)$};
        \node[font=\tiny] at (axis cs:2.5,1.5) {$R^\ast \in (0,1)$};
        \node[font=\tiny] at (axis cs:1.1,3) {$R^\ast \rightarrow 1$};
        \node[font=\tiny] at (axis cs:1.1,0.4) {$R^\ast \rightarrow 0$};
    \end{axis}
\end{tikzpicture}
  \flushleft\footnotesize \textit{Notes:} This figure illustrates the parameter regimes of Proposition~\ref{prop:benchmark_characterization} for $Q=0.05$. The dotted and dashed curves correspond to $\sigma=(a-Q(a+1))/(1-Q(a+1))$ and its reciprocal, respectively. An interior solution arises between the two curves.
  Otherwise, extreme rating standards are optimal.
\end{figure}
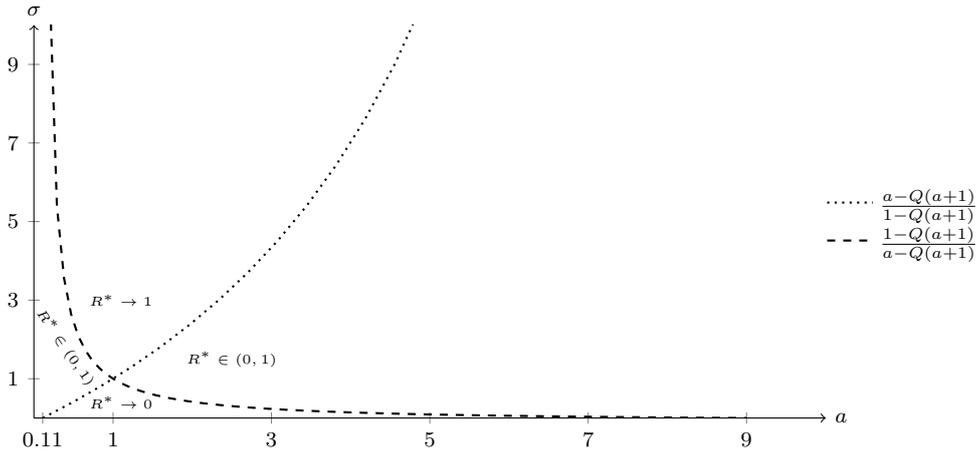

\section{Sales Maximization}\label{sec:sales}

We now consider a platform that maximizes sales of the product with a rating rather than receiver welfare. It will turn out that there are parameter regions where the two objectives are aligned and others where they are in stark conflict.

\subsection{Characterization} \label{sec:demand}
Using Corollary~\ref{cor:follow}, the sales for the rated product $D(R)$ can be expressed analogously to the value of the rating system in Proposition \ref{prop:value_of_RS} as 

  \begin{tabular}{lll}
  (i)   & $\pi^L(R)$ & if  $\vert\Delta_S^L\vert\leq 2\Delta_O^L$; \\
  (ii)  & $\pi^L(R)F(\tilde{i}(R))+(1-\pi^L(R))(1-F(\tilde{i}(R)))$ & if  $\Delta_S^L>2\Delta_O^L$; \\
  (iii) & $\pi^L(R)(1-F(\tilde{i}(R)))+(1-\pi^L(R))F(\tilde{i}(R))$ & if  $\Delta_S^L<-2\Delta_O^L$,
  \end{tabular}
  \smallskip
  
\noindent where we have made explicit that the indifferent receiver $\tilde{i}$ and the probability of a like $\pi^L$ are functions of the rating standard $R$. Similar to the expression for the value in Proposition \ref{prop:value_of_RS}, the relation between $\Delta_S^L$ and $\Delta_O^L$ determines which types are obedient. In case (i), all types are obedient so that sales are given by $\pi^L(R)$, the probability of a like. In cases (ii) and (iii), sales come from both obedient and disobedient types. Importantly, disobedient types buy a product rated with a dislike, thus contributing to sales.

The structural parallel between $D(R)$ above and the value expression in Proposition~\ref{prop:value_of_RS} is instructive. The determinants of both expressions are the same ($\pi^L$, $\Delta_O^L$, $\Delta_S^L$ and $\tilde{i}$). The key difference lies in the aggregation across receivers. Value weights each rating by the resulting payoff effect for the receiver, while demand simply counts purchases. When all types are obedient, sales equal $\pi^L$, the frequency of likes, regardless of their objective effect on receivers. This difference drives the potential conflict between value and sales maximization.

\subsection{Design}\label{sec:sales_design}

Proceeding as in Section~\ref{sec:value_design_sub}, we first show that, for symmetric populations, sales are monotone in the rating standard. Then, we establish an analogous result that applies to asymmetric populations as well.

\begin{proposition}\label{prop:demand_symmetric}
    For a decision environment $\mathcal{E}$ satisfying Assumption~\ref{ass:symmetry}, demand $D(R)=\pi^L(R)$ is decreasing in $R$.
\end{proposition}
The simple form of $D(R)$ is implied by Lemma~\ref{prop:pop_symmetry_allfollow}, which states that receivers are universally obedient under Assumption~\ref{ass:symmetry}. This implies the result because the probability of likes is decreasing in $R$.\footnote{See Lemma~\ref{lemma:solution_symmetric_population} in the appendix.}

Comparison with Proposition~\ref{prop:solution_symmetric_population} highlights the potential for conflict between value and sales maximization. When $\sigma < 1$, value is also decreasing in $R$, so the two objectives are aligned and both call for a low standard. When $\sigma > 1$, however, value is increasing in $R$ while sales remain decreasing, so that the objectives are diametrically opposed: Sales maximization favors a low standard to make likes frequent, while value maximization favors a high standard to make likes informative.

Like value maximization, sales maximization becomes more complex for asymmetric population distributions. Nonetheless, we establish a sufficient condition for sales maximization to call for low rating standards, mirroring the result in the symmetric case.

\begin{proposition} \label{prop:demand_asymmetric}
Consider a decision environment $\mathcal{E}$ with $q_g\geq q_b$. Then, the continuous extension of $D(R)$ to $[0,1]$ is maximized at $R=0$.
\end{proposition}

The result holds without assumptions on the distribution $F$, but requires $q_g \geq q_b$. Intuitively, sales come from obedient types who received a like and disobedient types who received a dislike. Lowering $R$ enlarges the like channel and shrinks the dislike channel in any decision environment. The assumption $q_g \geq q_b$ ensures the first effect dominates, making $R=0$ optimal.
This contrasts sharply with value maximization. In Example~\ref{ass:asymmetry} with $q_g > q_b$, the condition for value to be decreasing in $R$ in Proposition~\ref{prop:benchmark_characterization} is never satisfied: Since good products are common, likes are frequently given for low $R$, and the emphasis shifts from frequency to informativeness. Increasing $R$ filters out controversial products, making likes more selective and thus more valuable.

When \(q_g<q_b\), the analytical argument for Proposition~8 no longer
extends. One sufficient condition under which demand is nonetheless highest for low rating standards is \textit{universal obedience}: 
Then demand coincides with the probability of a like, so lowering \(R\)
weakly increases sales regardless of whether \(q_g\) is larger or smaller
than \(q_b\). Thus, the failure of Proposition~8 in the unrestricted model
must rely on some receivers being \textit{disobedient}. This is not merely a
limitation of the proof: without further restrictions on \(F\), one can
construct absolutely continuous full-support type distributions for which the
continuous extension of \(D(R)\) is not maximized at \(R=0\). The reason is
that, when bad products are sufficiently common, dislikes are frequent. If
the type distribution is sufficiently biased, the subjective effect may
dominate the objective effect for many receivers. Demand can then be driven
by disobedient receivers rather than by the frequency of likes.

This possibility appears not to arise in the tractable specification of
Example \ref{ass:asymmetry}, in which \(q_g<q_b\) corresponds to \(\sigma<1\). An
extensive numerical search found no violations of monotonicity, so that sales remain
decreasing in \(R\). Thus, while
Proposition 8 is false in the unrestricted model without \(q_g\ge q_b\), the
low-standard conclusion appears robust within Example \ref{ass:asymmetry}.

\section{Further Results} \label{sec:extensions}

\subsection{Optimal Pricing}\label{sec:price_extension_sym}

We now introduce prices to highlight that, when ratings are based on net utility, the price of a product affects not only its sales but also the value of a rating. This connects to a literature on how prices and ratings interact.\footnote{\cite{JohnenNg2024} show that firms can strategically use pricing to harvest favorable ratings, \cite{PeitzSobolev2024} document rating inflation, and \cite{carnehl2024pricing} study dynamic pricing when prices directly shape the distribution of reviews.} Our purpose is not to study full-fledged pricing competition, but to isolate, in the simplest setup, how the firm's pricing choice feeds back into the value of the rating.

We consider a symmetric population and assume for simplicity that types are distributed uniformly.\footnote{Assuming the uniform distribution simplifies the analysis and admits closed-form solutions. Economically, the key assumption is symmetry as in Assumption \ref{ass:symmetry} to ensure that the two controversial variants generate likes with the same probability at a given rating standard.} To simplify the exposition further, we assume that the controversial variants are equally likely and set $Q=q_1=q_2$. Receivers choose between buying the product with a rating at price $P\ge 0$ and buying the product without a rating at a price normalized to be $P_0=0$. The rating standard $R\in(0,1)$ is fixed (e.g., by a platform), and the firm selling the product with a rating chooses $P$ to maximize profit $\Pi(P;R)=P\,D(P;R)$, where $D(P;R)$ denotes demand at price $P$ given standard $R$.
Ratings are based on \emph{net} utility. For instance, after consuming product variant $c_1$, a sender gives a like if and only if $1/2+i-P \ge R$, resulting in an effective rating standard $T:=R+P$. We work with the continuous extension of the rating rule to $T=1$ and can therefore consider $T\in(0,1]$, equivalently prices $P\in[0,1-R]$, without loss. 

The receiver buys the product with a like if and only if 
\begin{align*}
    U^L(R+P)-P\ge U^0 \Leftrightarrow\ P\le \Delta_O^L(R+P).
\end{align*}
Since $\Delta_S^r=0$, all types consider a like objectively good news and the dislike objectively bad news. After a dislike rating, the product with a rating is dominated by the outside option.\footnote{Formally,
$U^D(R+P)\le U^0$ and $P\ge 0$ imply $U^D(R+P)-P\le U^0$.} Moreover, either all types buy the product with a like at price $P$ or none do. Therefore, sales at price $P$ take the simple form
\begin{equation}\label{eq:demand_sym_price}
D(P;R)=\pi^L(R+P)\cdot \mathbf 1\!\left\{P\le \Delta_O^L(R+P)\right\}, 
\end{equation}
and the firm's profits are $\Pi(P;R) = P \cdot D(P;R).$ Intuitively, at a price of zero, all consumers who observe a like buy the rated product. By increasing the effective rating standard $T=R+P$ and thus decreasing $\pi^L(R+P)$, a higher price makes a like less likely. However, conditional on observing a like, the rating becomes more valuable because it is harder to obtain. This increases the receiver's willingness to pay after a like, so that $\Delta_O^L(R+P)$ is increasing in $P$.\footnote{This mechanism is related to the rating inflation studied in \cite{PeitzSobolev2024} and the price-rating dynamics in \cite{JohnenNg2024}, where firms strategically manipulate the informativeness of ratings through pricing.} In regions where the latter effect dominates, sales are increasing in $P$. They remain positive as long as $P\le \Delta_O^L(R+P)$.

Define $P^u(R)$ as the unrestricted maximizer of the objective $P\,\pi^L(R+P)$, ignoring the demand constraint $P\le \Delta_O^L(R+P)$. Denote the upper bound of the feasible price set as
$\bar P(R):=\max\left\{P\in[0,1-R]: P\le \Delta_O^L(R+P)\right\}$. We then obtain:

\begin{proposition}\label{prop:pricing}
Consider a decision environment $\mathcal{E}$ such that the type distribution $F$ is uniform, $q_1=q_2$ and $q_g>Q$.\footnote{The assumption $q_1=q_2$ eliminates the \emph{subjective} content of ratings and delivers the simple ``all-or-nothing'' demand in \eqref{eq:demand_sym_price}. If instead $q_1\neq q_2$, then a like typically also shifts beliefs \emph{between} the two controversial variants. Because receiver types value these two variants differently, the same rating is more attractive for some types than for others. Demand after $L$ then becomes \emph{type-dependent}, and potentially only a subset of consumers are obedient at a given price.} Then, for a fixed rating standard $R\in(0,1)$, the firm's profit-maximizing price is $P^*(R)=\min\left\{\,P^u(R),\,\bar P(R)\right\}.$
\end{proposition}

The above result also sheds light on the conflict between value and sales maximization.  A higher price raises the effective rating standard $T=R+P$ above $R$. By Proposition~\ref{prop:solution_symmetric_population}, when (objectively) good products are common, this upward shift moves $T$ in a value-enhancing direction: a higher positive price makes the rating more selective and thereby more valuable. Conversely, when good products are scarce, a higher price pushes $T$ away from the value-maximizing level. In this case, the firm's profit motive undermines the value of the rating system.

\subsection{Multiple Ratings}\label{sec:many_recommendation}

We now show that additional ratings fail to add value when controversial variants are equally likely, or when the outside option is so attractive or unattractive that the single-rating optimum already induces the same decision. Formally, we consider a rating system where the receiver observes an arbitrary large number of ratings ($b$ likes and $d$ dislikes, each independently drawn from a population of senders as in the benchmark model).\footnote{We derive the posteriors for general $(b,d)$ in Appendix~\ref{AppMultipleRecommendations}.} To highlight the limits of informational gains from additional ratings, we focus on the case of \textit{infinite learning,} where the receiver observes an unbounded number of independent ratings.

\begin{proposition}\label{prop:no_value_infinite_learning}
    In the model with multiple ratings, suppose the decision environment $\mathcal{E}$ satisfies Assumption~\ref{ass:symmetry} holds. Let $\lambda := q_1/q_2$. Then, infinite learning does not increase the value of the rating system relative to the optimally designed single-rating benchmark if and only if one of the following holds:
    \begin{align*}
        \text{(i)}\ \lambda = 1;\quad
       \text{(ii)}\ \lambda > 1\ \text{and}\ \sigma \notin \left(\frac{1}{\lambda}, \lambda\right);\quad
        \text{(iii)}\ \lambda < 1\ \text{and}\ \sigma \notin \left(\lambda, \frac{1}{\lambda}\right).
    \end{align*}
\end{proposition}

Intuitively, a single like is sufficient to rule out the objectively bad variant, and a single dislike rules out the objectively good variant. Hence, when the receiver obtains a mix of both, he infers that the product must be one of the two controversial variants. With infinitely many ratings, the receiver can perfectly infer whether the product is objectively good, bad, or controversial. As in the case with only one rating, receivers buy objectively good products and reject bad ones. Thus, any potential gain must come from improved decisions about controversial products. However, controversial products always generate mixed ratings. If the odds of  good products are low, the outside option is unattractive, and the receiver may prefer a controversial product even if it is only an imperfect match. If the odds of good products are high, the receiver is more willing to choose the outside option. In either extreme, the decision under infinite learning coincides with what could be achieved by a single rating from an optimally designed rating system. Only when the uncontroversial variants are similarly likely and the likelihoods of the controversial variants are not too similar does additional information add value.

Finally, note that, under Assumption~\ref{ass:symmetry}, providing more ratings does not help a sales-maximizing platform either. To achieve this goal, it suffices to choose a single rating standard close to $0$, which leads all senders to issue likes when the product is not bad. Since mixed ratings can only reduce the likelihood of purchasing controversial products, adding them weakly reduces sales.\footnote{The finding that additional ratings need not add value is related to \cite{ifrach2019bayesian}, who develop a Bayesian model of sequential learning from binary like/dislike reviews and show that a seller's optimal dynamic pricing can cause learning to fail with positive probability. Both findings illustrate that more signals need not raise welfare, though through different mechanisms.}

\subsection{Distinct Sender and Receiver Distributions} \label{sec:distinct_distributions}

Sender and receiver distributions may be distinct. For instance, as early adopters, senders may differ systematically from the broader population.\footnote{\cite{acemoglu2022learning} study how such endogenous selection into reviewing shapes the speed of Bayesian learning. \cite{li2008self} document how early buyers with systematically different preferences generate self-selection bias in reviews even when ratings are mechanically truthful, framed around the same objective-quality versus subjective-fit decomposition we use.} Our earlier results extend naturally to such a setting. Let the sender be drawn from $F$ and the receiver from $G$. If $F$ satisfies Assumption~\ref{ass:symmetry}, the results from Sections~\ref{sec:value_design} and \ref{sec:sales} remain valid. For example, generalizing Proposition~\ref{prop:solution_symmetric_population}, Proposition~\ref{prop:solution_distinct_population} in the appendix shows that, when $F$ is symmetric, all receiver types are obedient,  $V(R)$ remains monotone in $R$ and $\sigma$ continues to play the same role as before in determining the sign of the slope. Similarly, the potential conflict between value and sales objectives persists. 

Beyond such robustness results, differing sender and receiver distributions give rise to new design considerations. In many cases, platforms may have partial control over the sender population or the information receivers have about it. For instance, Amazon's phrasing ``customers who bought this item also bought...'' implicitly informs the receiver about who the senders are. 

First, though it may seem that making the sender distribution more precise would raise the value of the system, this is not always true. A more concentrated sender pool makes ratings less dependent on idiosyncratic taste and therefore strengthens their objective content while reducing subjective content. Whether this raises or lowers value depends on which source of information is more important in the environment.\footnote{In particular, adapting the argument behind Proposition~\ref{prop:solution_distinct_population} shows that the effect depends on the prevalence of controversial products and on the extent to which receiver preferences are tilted toward one controversial variant or the other.}

Second, suppose that the sender and receiver distributions are fixed, but that the designer can provide the receivers with a signal about the sender's type. Assume for simplicity that the signal fully reveals the type. If the sender is revealed to be $i=0$, all types are obedient, as her rating only contains objective content. In contrast, if the signal reveals a sender of type $i\neq 0$, the receiver may learn about controversial products, too.\footnote{For instance, if $i=1/2$ and $R\in(0,1)$, a like reveals that the product is either of variant $g$ or $c_1$.} Thus, having a biased sender is not necessarily a bad thing.

Third, suppose sender and receiver types are correlated, so that ratings are effectively ``tailored'' to the receiver. Then the relevant sender distribution depends on the receiver's type. If these conditional sender distributions are symmetric, the subjective component of a rating vanishes for every receiver. Hence all receivers are obedient, ratings are interpreted purely through their objective content, and a sales-maximizing platform again prefers a rating close to $R=0$. 

These three cases share a common thread: The sender's characteristics shape not only how often likes are issued but also what they reveal. A biased sender, for instance, can be more informative about controversial products precisely because her ratings carry subjective content. The design of the sender pool thus offers a complementary lever to the rating standard for balancing objective and subjective content. \cite{bohren2025beyond} show experimentally that this lever can be relevant, where welfare in horizontally differentiated markets is restored by segmenting the rating system by consumer type.\footnote{Moreover, when the platform can influence the sender population, the choice of senders may itself depend on the objective: A value-maximizing platform may prefer a different sender composition than a sales-maximizing one, adding a new dimension to the value--sales conflict.}

\subsection{Polarization and the Prevalence of Controversial Products}\label{sec:comparative_statics}

Invoking symmetry (Assumption~\ref{ass:symmetry}), we now ask how the value of the rating system responds to two features of the environment other than the rating standard: the degree of preference polarization and the prevalence of controversial product variants.

First consider polarization (or preference heterogeneity). Recall that $\beta = F(1/2 - R)$ is the fraction of the population willing to recommend a controversial product at standard $R$. Under Assumption~\ref{ass:symmetry}, $\beta$ summarizes how often controversial products generate positive ratings. A higher $\beta$ raises the probability of likes, but weakens their objective content, since more positive ratings stem from controversial rather than uncontroversially good products. Whether greater preference heterogeneity raises or lowers the value of the rating system depends on $\sigma$: When uncontroversially good products are common, highly informative likes increase the value of the system, so a higher $\beta$ is harmful; when uncontroversially bad products are common, the value comes from ruling out bad products through more frequent likes, so a higher $\beta$ is beneficial (see Appendix~\ref{App:comparative_statics}). 

Next, we turn to the prevalence of controversial product variants. A natural conjecture is that a higher share of controversial products always reduces the value of the rating system, because it dilutes the objective content of recommendations. This is indeed true when good and bad products are similarly likely. More generally, however, the effect of $Q$ is non-monotone. Under Assumption~\ref{ass:symmetry}, the value of the rating system can be written as $V(R)=\pi^L\Delta_O^L$, reflecting that $Q$ has two opposing effects: Likes become more common but less likely to reflect uncontroversially good products. When $\sigma$ is far from one, the frequency effect can dominate the loss in objective informativeness, so the value of the rating system may be maximized at an interior share of controversial products (Corollary~\ref{cor:comparative_statics_Q} in Appendix~\ref{App:comparative_statics}).

\section{Conclusion}\label{sec:conclusion}

We analyze the design of coarse rating systems, focusing on the binary case of likes and dislikes.\footnote{One can show that the monotonicity result of Proposition \ref{prop:solution_symmetric_population} extends to richer rating structures. With three rating levels---like, neutral, dislike---defined by two rating standards $R_1 \leq R_2$, value remains monotone in the relevant standard under Assumption~\ref{ass:symmetry}: decreasing in $R_1$ if $\sigma < 1$ and increasing in $R_2$ if $\sigma > 1$. Intuitively, the neutral rating reveals that the product is controversial but carries no subjective content under symmetry, so the design logic reduces to pushing one of the extreme standards to its boundary.}  Our approach relies on the distinction between objective and subjective content of a rating. The payoff effect of a rating on a receiver has an objective part, shared by all types, and a subjective part, with sign depending on how closely her preferences are aligned with the sender's. A receiver is obedient if the objective effect outweighs the subjective one for her type. This decomposition shapes optimal design. When the population is sufficiently biased toward one of the controversial variants, a value-maximizing system preserves some subjective content at an interior standard, enabling inference about controversial products; otherwise, an extreme standard that filters out all subjective content is optimal.

The conflict between value and sales maximization is a central theme of our paper. When good products are scarce, both objectives call for low standards. When good products are common, this is usually still true for sales maximization but not for value maximization. Our result on pricing  reinforces this conflict. When good products are common, the firm's pricing enhances the value of the rating system by charging a high price, but when they are low, pricing undermines the value of the system.

Although disobedient receivers play an important conceptual role in the model, they appear only sparingly in the main design results. The reason is that optimal standards often induce obedience by all receiver types. Under sales maximization, this is often the case because the platform prefers extremely lenient standards, fostering obedience. Under value maximization, universal obedience always obtains for symmetric populations, and it can arise even under asymmetric populations when the subjective effect of ratings remains sufficiently weak. Thus, disobedient receivers are central for understanding how ratings can create value, but they are less often a feature of the optimally designed system itself.

There are several avenues for future research. For instance, the current model maintains classical assumptions such as risk neutrality and Bayesian updating. Relaxing these assumptions could affect which receiver types are obedient, with potential implications for the design of the system. On a related note, it would be interesting to use experimental evidence to see whether our model captures key features of receiver behavior and in which ways it fails to do so.

    \newpage
 \appendix
\setcounter{proposition}{0}
\setcounter{corollary}{0}
\setcounter{lemma}{0}

\renewcommand{\theproposition}{A.\arabic{proposition}}
\renewcommand{\thecorollary}{A.\arabic{corollary}}
\renewcommand{\thelemma}{A.\arabic{lemma}}

    \section{Formal Details and Proofs}

\label{SecAppProofs}

\subsection{Receiver Behavior}

\subsubsection{Posterior Distributions}\label{Sec:App_Posteriors}
The receiver's posterior $\mathbf{p}^L(R):= (p_g^L(R), p_1^L(R),p_2^L(R),p_b^{L}(R))$ after a like is
        \begin{align*}
        p_g^L(R) := \frac{q_g }{q_g+q_1\phi_1(R)+q_2\phi_2(R) }, \quad\quad   &p_2^L(R) :=  \frac{q_2\phi_2(R) }{q_g+q_1\phi_1(R)+q_2\phi_2(R) },\\
        p_1^L(R) := \frac{q_1\phi_1(R)}{q_g+q_1\phi_1(R)+q_2\phi_2(R)}, \quad \quad &p_b^{L}(R) := 0.     \end{align*}
A dislike gives rise to $\mathbf{p}^D(R):= (p_g^D(R), p_1^D(R),p_2^D(R),p_b^D(R))$, where
   \footnotesize \begin{align*}
        &p_g^D(R) := 0,\, &p_2^D(R) := \frac{q_2(1-\phi_2(R))}{q_1(1-\phi_1(R))+q_2(1-\phi_2(R))+ q_b},\\
        &p_1^D(R) := \frac{q_1(1-\phi_1(R))}{q_1(1-\phi_1(R))+q_2(1-\phi_2(R))+ q_b}, \,  &p_b^D(R) := \frac{q_b}{q_1(1-\phi_1(R))+q_2(1-\phi_2(R))+ q_b}. 
    \end{align*}\normalsize

    \subsubsection{Proof of Proposition \ref{prop:equivalence_following}}

Let $\Delta^{L}(i)\equiv U^{L}(i)-U^{0}(i)$ be the gain from buying after a like and $\Delta^{D}(i)\equiv U^{0}(i)-U^{D}(i)$ the gain from not buying after a dislike. The receiver buys after a like (does not buy after a dislike) if and only if $\Delta^{L}(i)\geq 0$ (resp.\ $\Delta^{D}(i)\geq 0$).

We first show the equivalence of these conditions. For every state $s\in\{g,1,2,b\}$, Bayes' rule implies $ \pi^{L}(p^{L}_{s}-q_{s})+\pi^{D}(p^{D}_{s}-q_{s})=0.$ Multiplying by the state-contingent payoffs $\{1, 1/2+i, 1/2-i, 0\}$ and summing over the states yields $\pi^{L} \Delta^{L}(i)=\pi^{D}\Delta^{D}(i).$ Because $\pi^{L},\pi^{D}>0$, $\Delta^{L}(i)$ and $\Delta^{D}(i)$ have the same sign, the desired equivalence holds.

It remains to show that this is equivalent to $\Delta_O^L\geq i\Delta_S^L$. By condition~\eqref{eq:optimal_following}, the receiver buys after a like if and only if
\begin{align*}
    (p_g^L-q_g) +\frac{p_1^L-q_1}{2} + \frac{p_2^L- q_2}{2} \;\geq\; i\bigl[(p_2^L-q_2)-(p_1^L-q_1)\bigr],
\end{align*}
which is exactly $\Delta_O^L\geq i\Delta_S^L$ by Definition~\ref{def:effects}.

    \subsubsection{Proof of Proposition \ref{prop:extreme_recommendations}}
First, observe that $    \lim_{R\rightarrow 1} p_1^L = \lim_{R\rightarrow 1} p_2^L = 0 $ and thus $\lim_{R\rightarrow 1} p_g^L= 1$. Hence, in the case $R\rightarrow 1$ the obedience condition   \eqref{eq:optimal_following} simplifies to
\begin{align}\label{eq:R1}
    1-q_g-\frac{q_1+q_2}{2} \geq i(q_1-q_2).
\end{align}
Equation (\ref{eq:R1}) always holds for $i\geq 0$ and $q_1 \ge q_2$. Further, for $q_1>q_2$ the condition is most stringent for $i=1/2$. In this case, it is equivalent to $1 \geq q_1 + q_g$, which always holds. The case $i<0$ is analogous. In the case $R\rightarrow 0$ we get $ \lim_{R\rightarrow 0} p_s^L = {q_s}/(1-q_b)$ for $s\in \{g, 1, 2\}$ and thus the obedience condition simplifies to
\begin{align*}
  &  \frac{q_gq_b}{1-q_b} +\frac{q_1q_b}{2(1-q_b)} +\frac{q_2q_b}{2(1-q_b)}  \geq i[\frac{q_2q_b}{1-q_b} -\frac{q_1q_b}{1-q_b} ]\\
     &  \Leftrightarrow 2q_g +q_1 +q_2\geq 2i[q_2-q_1 ].
\end{align*}
For $i\geq 0$, the last condition always holds if $q_1\geq q_2$. Further, for $q_2>q_1$ the condition is most stringent for $i=1/2$. In this case, it is equivalent to $q_g +  q_1 \geq 0$, which is always satisfied. The case $i<0$ is analogous.


\subsection{Value of Rating Systems}

    \subsubsection{ Proof of Proposition \ref{prop:value_of_RS}}

\noindent (i) If $|\Delta_S^L|\leq 2 \Delta_O^L$, inserting (\ref{eq_VNR}) and (\ref{eq_VF}) yields
\begin{small}
     \begin{align*}
      V(R) &= \int_{-1/2}^{1/2} (V^+(i) - V^0(i))dF(i)\\
      &= \pi^L\int_{-1/2}^{1/2} \bigl(p_g^L+(1/2+i) p_1^L+(1/2-i) p_2^L\bigr) - (q_g + (1/2+i) q_1 + (1/2-i)  q_2 ) dF(i)\\
       &= \pi^L\int_{-1/2}^{1/2} (p_g^L-q_g) +\frac{p_1^L-q_1}{2} + \frac{p_2^L- q_2}{2}dF(i)- \pi^L\int_{-1/2}^{1/2} i ([(p_2^L-q_2)-(p_1^L-q_1) ])dF(i)\\
              &= \pi^L\int_{-1/2}^{1/2} \underbrace{\frac{p_g^L-q_g}{2} -\frac{p_b^L-q_b}{2}}_{= \Delta_O^L}dF(i)- \pi^L\int_{-1/2}^{1/2} i (\underbrace{[(p_2^L-q_2)-(p_1^L-q_1) ]}_{= \Delta_S^L})dF(i)\\
       &= \pi^L  [\Delta_O^L - \Delta_S^L \E[i]]
  \end{align*}
  \end{small}
\noindent  (ii) If $\Delta_S^L<-2 \Delta_O^L$, inserting (\ref{eq_VNR}), (\ref{eq_VF}) and (\ref{eq_VNF}) yields
\begin{small}
      \begin{align*}
     & V(R)=\int_{-1/2}^{\tilde{i}}  V^-(i)dF(i) + \int_{\tilde{i}}^{1/2}V^+(i) dF(i) - \int_{-1/2}^{1/2} V^0(i)dF(i)\\
      &= \int_{-1/2}^{\tilde{i}} \pi^L\left(q_g+q_1 (\frac{1}{2}+i) +q_2 (\frac{1}{2}-i)\right)+ (1-\pi^L)\left(p_g^D+p_1^D (\frac{1}{2}+i) +p_2^D (\frac{1}{2}-i)\right) dF(i)\\
      &+ \int_{\tilde{i}}^{1/2} \pi^L\left(p_g^L+p_1^L (\frac{1}{2}+i) +p_2^L (\frac{1}{2}-i)\right)+ (1-\pi^L)\left(q_g+q_1 (\frac{1}{2}+i) +q_2 (\frac{1}{2}-i)\right)dF(i)\\
      &- \int_{-1/2}^{\tilde{i}} q_g+q_1 (\frac{1}{2}+i) +q_2 (\frac{1}{2}-i)dF(i)-\int_{\tilde{i}}^{1/2} q_g+q_1 (\frac{1}{2}+i) +q_2 (\frac{1}{2}-i)dF(i)\\
      &= (1-\pi^L)\int_{-1/2}^{\tilde{i}}\underbrace{\frac{p_g^D-q_g}{2} -\frac{p_b^D-q_b}{2}}_{= \Delta_O^D} - i(\underbrace{[(p_2^D-q_2)-(p_1^D-q_1) ]}_{= \Delta_S^D})dF(i)\\
      &+ \pi^L\int_{\tilde{i}}^{1/2} \underbrace{\frac{p_g^L-q_g}{2} -\frac{p_b^L-q_b}{2}}_{= \Delta_O^L} - i(\underbrace{[(p_2^L-q_2)-(p_1^L-q_1) ]}_{= \Delta_S^L}) dF(i)\\
      &= (1-\pi^L)F(\tilde{i}) \left[\Delta_O^D  - \Delta_S^D\E[i \mid i\leq \tilde{i}]\right]+ \pi^L(1- F(\tilde{i}))\left[\Delta_O^L - \Delta_S^L  \E[i \mid i\geq \tilde{i}]\right].
  \end{align*}
\end{small}
 
\noindent (iii) If $\Delta_S^L>2 \Delta_O^L$, then $V(R)=\int_{-1/2}^{\tilde{i}} V^+(i) dF(i) + \int_{\tilde{i}}^{1/2} V^-(i) dF(i) - \int_{-1/2}^{1/2} V^0(i)dF(i)$. Inserting (\ref{eq_VNR}), (\ref{eq_VF}) and (\ref{eq_VNF}) and proceeding analogously as in case (ii) gives the expression.

       \subsubsection{Proof of Lemma \ref{prop:pop_symmetry_allfollow}}
    According to Corollary \ref{cor:follow}, all types are obedient if 
    \begin{align}
       \min\{\Delta_O^L+\frac{\Delta_S^L}{2},\Delta_O^L-\frac{\Delta_S^L}{2}\}\geq 0. 
    \end{align} 
    To see that this condition is fulfilled, first observe that
    \begin{align*}
        \Delta_O^L +\frac{\Delta_S^L}{2}&=  p_g^L- q_g +  p_2^L- q_2\\
        &= \frac{q_g}{q_g + \beta(q_1+q_2)} - q_g + \frac{\beta q_2}{q_g + \beta(q_1+q_2)} - q_2 \\
        &= \frac{\beta(q_2 - (q_g + q_2)(q_1+q_2)) + q_g(1-q_g-q_2)}{q_g + \beta(q_1+q_2)}
    \end{align*}
    This term is positive if $q_2 - (q_g + q_2)(q_1+q_2)\geq 0$. If $q_2 - (q_g + q_2)(q_1+q_2)< 0$, then
 \begin{align*}
        &\beta(q_2 - (q_g + q_2)(q_1+q_2)) + q_g(1-q_g-q_2)\\
        &\geq (q_2 - (q_g + q_2)(q_1+q_2)) + q_g(1-q_g-q_2)=  (q_g+q_2)(1-q_1-q_2-q_g)\geq 0\text{.}
    \end{align*}
    Thus, $\Delta_O^L+\frac{\Delta_S^L}{2}\geq 0$ even in this case. Proceeding analogously, one obtains $\Delta_O^L-\frac{\Delta_S^L}{2}=p_g^L- q_g +  p_1^L- q_1\geq0.$ Corollary \ref{cor:follow} thus implies that all types are obedient. Assumption \ref{ass:symmetry} implies $\E[i]=0$, so that the result follows from Proposition \ref{prop:value_of_RS}.

\subsubsection{Proof of Proposition \ref{prop:solution_symmetric_population}}
We first restate the probabilities of likes, their objective effect and the value of the rating system in terms of $\sigma$, $\beta$ and $Q$:
\begin{lemma}\label{lemma:Reparameterization}
Suppose Assumption~\ref{ass:symmetry} holds. Then,
 \begin{align*}
            \pi^L &= \frac{(1 - 2 Q) \sigma}{1 + \sigma}+ 2 Q \beta; \\
            \Delta_O^L &=\frac{\beta Q (\sigma+1)-2 Q \sigma+\sigma}{2 \beta Q (\sigma+1)-2 Q \sigma+\sigma}- \frac{(1 - 2 Q) \sigma}{1 + \sigma} - Q;\\
            V(\beta) &= \pi^L\Delta_O^L=\notag (1-2Q)\frac{\sigma+Q\left(\beta-\sigma+\sigma^2(1-\beta)\right)}{(\sigma+1)^2}.\label{eq:symetric_value}
        \end{align*}
\end{lemma}
\begin{proof}
    The result on $\pi^L$ follows from inserting $q_g = (1 - 2 Q) \sigma/(1 + \sigma)$, $\phi_1(R)=\phi_2(R)=\beta$ and $ q_1+q_2=2Q$ into $\pi^L=q_g + q_1\phi_1(R)+q_2\phi_2(R)$.
Next, equation \eqref{eq:objective} gives
           \begin{align}\label{eq:ObjectiveApp}
               \Delta_O^L=
               \frac{\beta Q (\sigma+1)-2 Q \sigma+\sigma}{2 \beta Q (\sigma+1)-2 Q \sigma+\sigma}- \frac{(1 - 2 Q) \sigma}{1 + \sigma} - Q
           \end{align}
The result on $V(\beta)$ then follows immediately from Proposition \ref{prop:value_of_RS}.
\end{proof}
   
 The following result shows that countervailing effects of $\beta$ on the probability of a like and the objective value $\Delta_O^L$ of accepting it result in a non-trivial effect on $ V(\beta) =\pi^L\Delta_O^L$. 
    \begin{lemma} \label{lemma:solution_symmetric_population}
         Suppose Assumption~\ref{ass:symmetry} holds. Then,
        \begin{itemize}
            \item[(i)]$\pi^L$ is strictly increasing and $\Delta_O^L$ strictly decreasing in $\beta$;
         \item[(ii)] $\frac{\partial^2 V}{\partial \beta\partial \sigma}<0$ and $ V'(\beta)\text{ has the same sign as }1-\sigma$. 
         \end{itemize}
    \end{lemma}
\begin{proof}
   (i)    Note that $\pi^L= (1 - 2 Q) \sigma/(1 + \sigma)+ 2 Q \beta$, so that $\pi^L$ is strictly increasing.  Moreover, $\derivative{\Delta_O^L}{\beta} = -\frac{(Q (1 - 2 Q) \sigma (1 + \sigma))}{(\sigma - 2 Q \sigma + 2 \beta Q (1 + \sigma))^2}<0.$     (ii) Lemma \ref{lemma:Reparameterization} implies $V'(\beta) = \frac{Q(1-2Q)(1-\sigma)}{1 +\sigma},$ and hence $\derivative{^2V}{\beta\partial\sigma } =-\frac{2Q(1-2Q)}{(1 +\sigma)^2} <0,$ so that $V'(\beta)$ is negative for $\sigma>1$, positive for $\sigma<1$ and zero for $\sigma=1$.
\end{proof}

 Then, using $\beta = F(1/2-R)$, Lemma~\ref{lemma:solution_symmetric_population} directly implies Proposition \ref{prop:solution_symmetric_population}.\footnote{Result (iii) follows because $V(\beta)$ is independent of $\beta$ in this case.}

\subsubsection{Proof of Proposition \ref{prop:asymmetric_interior_unequal_q}}

The proof follows from Steps 1--6, which we now state and prove. We will repeatedly use the formulation that \textit{a statement holds near the extremes}, meaning that for any decision environment $\mathcal{E}$, there exists $\delta (\mathcal{E}) >0$ such that  the statement holds for $R< \delta (\mathcal{E}) $ and $R>1-\delta (\mathcal{E})$.

\textbf{Step 1}: \textit{Near the extremes},
$V(R)=\pi^L(R)\big(\Delta_O^L(R)-\mu\,\Delta_S^L(R)\big)$.\newline
\textbf{Proof of Step 1}: For  $R\to 0$ or $R\to 1$, all types have a strict preference for being obedient. Since $\Delta_O^L(R)$ and $\Delta_S^L(R)$ are continuous in $R$, for any decision environment $\mathcal{E}$, there exists
$\delta (\mathcal{E}) >0$ such that all types are obedient for all
$R\in(0,\delta (\mathcal{E}))\cup(1-\delta (\mathcal{E}),1)$.
Hence, for such $R$ Proposition~\ref{prop:value_of_RS}(i) applies, and Step 1 follows.

\textbf{Step 2}:  \textit{Near the extremes,}\newline
$V(R)=q_1\Big(\tfrac12+\mu\Big)\phi_1(R)+q_2\Big(\tfrac12-\mu\Big)\phi_2(R)-\left[q_g+\frac{q_1+q_2}{2}-\mu(q_2-q_1)\right]\,\pi^L(R)+C,$\newline
\textit{where the term $C$ is independent of $R$}.

\textbf{Proof of Step 2}: Recalling the posteriors after a like (see page~\pageref{SecAppProofs}) and applying Definition~\ref{def:effects}, we obtain 
\begin{align*}
\pi^L\Delta_O^L
&=\Big(q_g-\pi^L q_g\Big)+\frac{q_1\phi_1-\pi^L q_1}{2}+\frac{q_2\phi_2-\pi^L q_2}{2},\\
\pi^L\Delta_S^L
&=\Big(q_2\phi_2-\pi^L q_2\Big)-\Big(q_1\phi_1-\pi^L q_1\Big)
=(q_2\phi_2-q_1\phi_1)-\pi^L(q_2-q_1).
\end{align*}
Inserting these terms into the expression from Step 1 and rearranging gives the claim.

\textbf{Step 3}: \textit{Near the extremes, }$V'(R)=q_1A_1\,\phi_1'(R)+q_2A_2\,\phi_2'(R)$.\newline
\textbf{Proof of Step 3}: Differentiating the expression obtained in Step 2 gives
\begin{align*}
    V'(R)&=q_1\Big(\tfrac12+\mu\Big)\phi_1'(R)+q_2\Big(\tfrac12-\mu\Big)\phi_2'(R)\\
&-\left[q_g+\frac{q_1+q_2}{2}-\mu(q_2-q_1)\right]\big(q_1\phi_1'(R)+q_2\phi_2'(R)\big).
\end{align*}
The claim follows from this expression using
\begin{align*}
    &\tfrac12+\mu-\left[q_g+\frac{q_1+q_2}{2}-\mu(q_2-q_1)\right] = \frac{q_b-q_g}{2}+\mu(q_2-q_1)+\mu = A_1,\\
    &\tfrac12-\mu-\left[q_g+\frac{q_1+q_2}{2}-\mu(q_2-q_1)\right] = \frac{q_b-q_g}{2}+\mu(q_2-q_1)-\mu = A_2.
\end{align*}

\textbf{Step 4}: $V'(0+)=-q_1A_1 f_- - q_2A_2 f_+$ \textit{and}  $V'(1-)=-q_1A_1 f_+ - q_2A_2 f_-$.

\textbf{Proof of Step 4}: Since $\phi_1'(R)=-f(R-\tfrac12)$ and $\phi_2'(R)=-f(\tfrac12-R)$,
\[
V'(R)=-q_1A_1\,f(R-\tfrac12)-q_2A_2\,f(\tfrac12-R)
\]
near the extremes. Let $f_-:=f(-\tfrac12)$ and $f_+:=f(\tfrac12)$. 
Taking limits in the expression for $V'(R)$ gives the result.

\textbf{Step 5:}  \textit{Condition~\eqref{eq:asymmetric_endpoint_condition} is equivalent to $V'(0+)>0$ and $V'(1-)<0.$ }

\textbf{Proof of Step 5}: 
Since Assumption~\ref{ass:full_support_bounded_density} implies $f(\tfrac12)>0$, we have
\[
V'(0+)>0
\quad\Longleftrightarrow\quad
q_1A_1 f(-\tfrac12)+q_2A_2 f(\tfrac12)<0,
\]
and
\[
V'(1-)<0
\quad\Longleftrightarrow\quad
q_1A_1 f(\tfrac12)+q_2A_2 f(-\tfrac12)>0.
\]
Thus condition~\eqref{eq:asymmetric_endpoint_condition} is equivalent to $V'(0+)>0$ and $V'(1-)<0$.

\textbf{Step 6:} \textit{$V(R)$ has an interior maximum $R^\ast\in(0,1)$.}\newline
\textbf{Proof of Step 6}: According to Step 5, condition~\eqref{eq:asymmetric_endpoint_condition} implies $V'(0+)>0$ and $V'(1-)<0$.
Hence, for any decision environment $\mathcal{E}$, there exists $\delta(\mathcal{E})>0$ such that $V$ is strictly increasing on $(0,\delta(\mathcal{E}))$ and strictly decreasing on $(1-\delta(\mathcal{E}),1)$. Thus the supremum of $V$ over $(0,1)$ equals the supremum of $V$ over some compact interval $[\delta(\mathcal{E}),1-\delta(\mathcal{E})]\subset(0,1)$. Since $V$ is continuous in $R$,\footnote{Since payoffs are bounded and, for every type $i$, the receiver's ex ante payoff under standard $R$ is a continuous function of $R$, the value function $V(R)$ is continuous by the dominated convergence theorem.} the extreme value theorem implies that $V(R)$ attains its maximum on $[\delta(\mathcal{E}),1-\delta(\mathcal{E})]$, and hence on $(0,1)$, at some $R^\ast\in(0,1)$.

\subsubsection{Proof of Proposition \ref{prop:benchmark_characterization}}
    Throughout this proof, we use the identities $q_g=(1-2Q)\sigma/(1+\sigma)$ and $q_b=(1-2Q)/(1+\sigma)$. Since Proposition~\ref{prop:benchmark_characterization} assumes $Q\leq \bar Q(a,\sigma)$, Lemma~\ref{lemma:universal_acceptance_cutoff} implies that all receiver types are obedient for every $R\in(0,1)$. Hence Proposition~\ref{prop:value_of_RS}(i) applies throughout, yielding $V(R) = \pi^L[\Delta_O^L - \Delta_S^L \E[i]].$   In Example~\ref{ass:asymmetry}, $\phi_1(R)=1-R^a$, $\phi_2(R)=(1-R)^a$, and $\E[i]=\frac{a-1}{2(a+1)}$, so that
    \begin{align*}
        V(R)
        &=  \frac{(1-2 Q) \sigma (Q (\sigma-1)+1)}{(\sigma+1)^2} + \frac{Q}{a+1} (1-R^a) \left(a - (a+1)\left(\frac{(1-2Q)\sigma}{1 +\sigma} +Q\right)\right)\\
        & + \frac{Q}{a+1} (1-R)^a \left(1 - (a+1)\left(\frac{(1-2Q)\sigma}{1 +\sigma} +Q\right)\right)\\
        &=c_0 + \frac{Q}{a+1}\left[c_1(a)(1-R^a)+c_2(a)(1-R)^a\right]
    \end{align*}
    with $c_0= \frac{(1-2 Q) \sigma (Q (\sigma-1)+1)}{(\sigma+1)^2} $, $c_1(a) = a  +(a+1) \frac{Q( \sigma-1)-\sigma}{\sigma+1}$ and $c_2(a) =1 + (1+a)\frac{Q(\sigma-1)-\sigma}{\sigma+1}$. We thus obtain
    \begin{align*}
        V'(R) &= -\frac{Qa}{a+1}\left[R^{a-1} c_1(a)+(1-R)^{a-1}c_2(a)\right],\\
        V''(R) &= -\frac{Qa(a-1)}{a+1}\left[R^{a-2} c_1(a)-(1-R)^{a-2}c_2(a)\right],
    \end{align*}
    Since $R^{a-1}>0$ and $(1-R)^{a-1}>0$ on $(0,1)$ and $Qa/(a+1)>0$, the sign of $V'(R)$ is determined by the signs of $c_1(a)$ and $c_2(a)$, which can be written as
    \begin{align*}
       c_1(a)=&\frac{\sigma[Q(a+1)-1]+a-Q(a+1)}{\sigma+1},\\
      c_2(a)=&\frac{\sigma[Q(a+1)-a]+1-Q(a+1)}{\sigma+1}.
    \end{align*}
Solving $c_1(a)=0$ and $c_2(a)=0$ gives $\sigma_1^\ast := \frac{a-Q(a+1)}{1-Q(a+1)}$ and $\sigma_2^\ast := \frac{1-Q(a+1)}{a-Q(a+1)}.$ If $Q<\min\{a,1\}/(a+1)$, then the numerators of $c_1(a)$ and $c_2(a)$ are strictly decreasing affine functions of $\sigma$ with positive intercepts. Therefore
\[
c_1(a),c_2(a)\geq 0 \iff \sigma\leq \min\{\sigma_1^\ast,\sigma_2^\ast\},
\]
\[
c_1(a),c_2(a)\leq 0 \iff \sigma\geq \max\{\sigma_1^\ast,\sigma_2^\ast\}.
\]
Hence, if $Q<\min\{a,1\}/(a+1)$, $V$ is decreasing exactly if $\sigma\leq \min\{\sigma_1^\ast,\sigma_2^\ast\}$ and increasing exactly if $\sigma\geq \max\{\sigma_1^\ast,\sigma_2^\ast\}$.

If instead $Q\geq \min\{a,1\}/(a+1)$, these two monotone cases cannot occur. To see this, first suppose $a>1$, so that the cutoff equals $1/(a+1)$. Since $Q<1/2<a/(a+1)$, we have $Q(a+1)-1\geq 0$ and $a-Q(a+1)>0,$ and hence $c_1(a)>0$ for every $\sigma>0$. Moreover, $Q(a+1)-a<0$ and $1-Q(a+1)\leq 0,$ so that $c_2(a)<0$ for every $\sigma>0$. Thus $c_1(a)>0>c_2(a)$.

Now suppose $a<1$, so that the cutoff equals $a/(a+1)$. Since $Q<1/2<1/(a+1)$, we have $Q(a+1)-1<0$ and $a-Q(a+1)\leq 0$ and hence $c_1(a)<0$ for every $\sigma>0$. Moreover, $Q(a+1)-a\geq 0,$ and $1-Q(a+1)>0,$ so $c_2(a)>0$ for every $\sigma>0$. Thus $c_1(a)<0<c_2(a)$. Therefore, at and above the cutoff $\min\{a,1\}/(a+1)$, the coefficients have opposite signs for every $\sigma>0$. Together with the case $Q<\min\{a,1\}/(a+1)$ and
\[
\min\{\sigma_1^\ast,\sigma_2^\ast\}<\sigma<\max\{\sigma_1^\ast,\sigma_2^\ast\},
\]
this is exactly the  region where $c_1(a)$ and $c_2(a)$ have opposite signs. Since
\begin{align*}
c_1(a)-c_2(a)
\begin{cases}
<0 & \text{if } a<1,\\
>0 & \text{if } a>1,
\end{cases}
\end{align*}
the mixed-sign case means $c_1(a)>0>c_2(a)$ when $a>1$, and $c_1(a)<0<c_2(a)$ when $a<1$. In either case,
\begin{align*}
V''(R)=-\frac{Qa(a-1)}{a+1}\bigl[R^{a-2}c_1(a)-(1-R)^{a-2}c_2(a)\bigr]<0 \qquad\text{for all } R\in(0,1),
\end{align*}
so $V$ is strictly concave on $(0,1)$. To conclude that the maximum is interior, we also check the derivatives at $R=0$ and $R=1$. If $a>1$ and $c_1(a)>0>c_2(a)$, then
\[
V'(0+)=-\frac{Qa}{a+1}c_2(a)>0,\qquad V'(1-)=-\frac{Qa}{a+1}c_1(a)<0.
\]
If $a<1$ and $c_1(a)<0<c_2(a)$, then $R^{a-1}c_1(a)\to -\infty$ as $R\downarrow 0$ while $(1-R)^{a-1}c_2(a)\to c_2(a)$, so $V'(0+)>0$; similarly, $(1-R)^{a-1}c_2(a)\to +\infty$ as $R\uparrow 1$, so $V'(1-)<0$. Hence $V$ increases near $0$ and decreases near $1$, and strict concavity implies a unique interior maximizer $R^\ast\in(0,1)$.

Combining the cases $Q\geq \min\{a,1\}/(a+1)$ and $Q<\min\{a,1\}/(a+1)$ proves
  Proposition~\ref{prop:benchmark_characterization}.


\subsection{Sales Maximization}

\subsubsection{Proof of Proposition \ref{prop:demand_asymmetric}}

We show $D(R)\le D(0)$ for all $R>0$. Let $\hat{F}(R)\in[0,1]$ denote the fraction of receiver types that buy a product with a like at standard $R$. Then the claim is equivalent to
\begin{equation}\label{eqAux}
    \pi^L(R)\hat{F}(R)+ (1-\pi^L(R))(1- \hat{F}(R)) \leq q_g+q_1+q_2 .
\end{equation}
Fix $R>0$. Then the left-hand side of
\eqref{eqAux} equals
\[
\pi^L(R)\hat{F}(R)+(1-\pi^L(R))(1-\hat{F}(R))=(1-\pi^L(R))+(2\pi^L(R)-1)\hat{F}(R).
\]
As a function of $\hat{F}(R)\in[0,1]$ this expression is affine, hence it is bounded above by its
endpoint values:
\[
\pi^L(R) \hat{F}(R)+(1-\pi^L(R))(1-\hat{F}(R))\le \max\{\,1-\pi^L(R),\ \pi^L(R)\,\}.
\]
Next, by definition $\pi^L(R)=q_g+q_1\phi_1(R)+q_2\phi_2(R)$ with $\phi_1,\phi_2\in[0,1]$, so
\[
\pi^L(R)\le q_g+q_1+q_2 = D(0).
\]
Moreover, since $\phi_1,\phi_2\ge 0$, we also have $\pi^L(R)\ge q_g$. Under the assumption
$q_g\ge q_b$ and using $q_g+q_1+q_2=1-q_b$, it follows that
\[
1-\pi^L(R)\le 1-q_g \le 1-q_b = q_g+q_1+q_2 = D(0).
\]
Therefore both $\pi^L(R)$ and $1-\pi^L(R)$ are at most $D(0)$, implying
\[
\max\{1-\pi^L(R),\pi^L(R)\}\le D(0).
\]
Combining the inequalities yields \eqref{eqAux}, hence $D(R)\le D(0)$ for all $R>0$.
Consequently, sales are maximized as $R\to 0$.

\subsection{Further Results}

\subsubsection{Optimal Pricing}

\subsubsection*{Proof of Proposition \ref{prop:pricing}}

The proof will follow from Steps 1-4 below.

\smallskip
\noindent\textbf{Step 1.} \textit{The unique unrestricted profit maximizer is}
\[
P^u(R)=\frac{q_g+2Q(1-R)}{4Q}.
\]

Fix $R\in(0,1)$ and consider prices $P\in[0,1-R]$. For a uniform type distribution and $q_1=q_2=Q$, a sender gives a like for each controversial variant with probability $1-T=1-R-P$. Hence
\[
\pi^L(R+P)=q_g+q_1(1-R-P)+q_2(1-R-P)=q_g+2Q(1-R-P).
\]
Ignoring the demand constraint in \eqref{eq:demand_sym_price}, the firm's (unrestricted) profit is
\[
\Pi^u(P;R):=P\,\pi^L(R+P)=P\big(q_g+2Q(1-R-P)\big).
\]
This function is strictly concave in $P$ (since $Q>0$), and its first-order condition is
\[
\frac{\partial \Pi^u(P;R)}{\partial P}
=q_g+2Q(1-R-P)-2QP
=q_g+2Q(1-R)-4QP=0.
\]
This first-order condition yields
\[
P^u(R)=\frac{q_g+2Q(1-R)}{4Q}.
\]
This is the unconstrained maximizer of the quadratic objective $\Pi^u(P;R)$. It need not lie in the admissible interval $[0,1-R]$, so at this stage it is not yet the solution to the constrained pricing problem.

\smallskip
\noindent\textbf{Step 2.} \textit{The objective content of the rating system satisfies}
\[
\Delta_O^L(R+P)=\frac12-q_g-Q+\frac{q_g}{2\big(q_g+2Q(1-R-P)\big)}.
\]

Sales are positive if and only if the acceptance constraint in \eqref{eq:demand_sym_price} holds, i.e.\ if and only if $P\le \Delta_O^L(R+P).$  Under $q_1=q_2=Q$ and uniform types, $\Delta_S^L=0$, so $\Delta_O^L(R+P) = U^L(i) - U^0(i)$ is the common willingness to pay after a like for all types $i$. Moreover, the posterior probability of a good product after $L$ is
\[
p_g^L(R+P)=\frac{q_g}{q_g+2Q(1-R-P)},
\]
and the posterior probabilities of the two controversial variants are equal. Plugging these posteriors into Definition~\ref{def:effects}(i) yields the result.

\textbf{Step 3}:\textit{ The acceptance gap} $\gamma(P):=\Delta_O^L(R+P)-P$ \textit{ is strictly decreasing on} $[0,1-R]$.

Differentiating $\gamma$ gives
\[
\gamma'(P)=\frac{q_gQ}{\big(q_g+2Q(1-R-P)\big)^2}-1.
\]
For all $P\in[0,1-R]$ we have $q_g+2Q(1-R-P)\ge q_g$, hence
\[
\frac{q_gQ}{\big(q_g+2Q(1-R-P)\big)^2}\le \frac{q_gQ}{q_g^2}=\frac{Q}{q_g}<1,
\]
where the strict inequality uses the assumption $q_g>Q$. Therefore $\gamma'(P)<0$ on $[0,1-R]$.

Moreover, $\gamma(0)=\Delta_O^L(R)>0$. To see this, note that
\[
\Delta_O^L(R)=\frac{p_g^L(R)-q_g+q_b}{2}.
\]
Since $R\in(0,1)$ and $Q>0$, we have $q_g+2Q(1-R)<q_g+2Q=1-q_b<1$, so $p_g^L(R)=q_g/(q_g+2Q(1-R))>q_g$. Hence $p_g^L(R)-q_g+q_b>0$, proving $\gamma(0)>0$. Therefore the feasible set
\[
\Big\{P\in[0,1-R]:\; P\le \Delta_O^L(R+P)\Big\}
=\Big\{P\in[0,1-R]:\; \gamma(P)\ge 0\Big\}
\]
Since $\gamma$ is continuous and strictly decreasing, this feasible set is a nonempty compact interval of the form $[0,\bar P(R)]$, where $\bar P(R):=\max\{P\in[0,1-R]: \gamma(P)\ge 0\}$. If $\gamma(1-R)\le 0$, then $\bar P(R)$ is the unique root of $\gamma$ in $[0,1-R]$, characterized by $\gamma(\bar P(R))=0$, i.e.\ $\bar P(R)=\Delta_O^L(R+\bar P(R))$. If instead $\gamma(1-R)>0$, then $\bar P(R)=1-R$ and every price in $[0,1-R]$ is feasible.

\textbf{Step 4}:\textit{ The unique maximizer on} $[0,\bar P(R)]$ \textit{is} $P^*(R)=\min\{P^u(R),\bar P(R)\}.$

From Step 3, the original problem reduces to maximizing the strictly concave quadratic $\Pi^u(P;R)$ over the interval $[0,\bar P(R)]$. Since $P^u(R)$ is the unconstrained maximizer from Step 1, the unique constrained maximizer is
\[
P^*(R)=
\begin{cases}
P^u(R), & \text{if } P^u(R)\le \bar P(R),\\[4pt]
\bar P(R), & \text{if } P^u(R)>\bar P(R).
\end{cases}
\]
Equivalently,
\[
P^*(R)=\min\{P^u(R),\bar P(R)\}.
\]
Since $\bar P(R)\le 1-R$, this formula automatically respects the admissible price restriction $P\in[0,1-R]$.

\subsubsection{Multiple ratings}
\label{AppMultipleRecommendations}

We now derive some auxiliary results for the case of multiple ratings, before we prove Proposition \ref{prop:no_value_infinite_learning}.

After $b$ likes and $d$ dislikes, where $b+d>0$, the posterior reads
\begin{align*}
    p_g(b,d) &= \begin{cases}
        \frac{q_g}{q_g + q_1(\phi_1(R))^b + q_2(\phi_2(R))^b} & \text{if } d=0\\
        0 & \text{if } d>0\\
    \end{cases}\\
    p_1(b,d) &= \frac{q_1(\phi_1(R))^b (1-\phi_1(R))^d }{q_1(\phi_1(R))^b (1-\phi_1(R))^d  + q_2(\phi_2(R))^b (1-\phi_2(R))^d  }\\
	p_2(b,d)&= \frac{q_2(\phi_2(R))^b (1-\phi_2(R))^d  }{q_1(\phi_1(R))^b (1-\phi_1(R))^d  + q_2(\phi_2(R))^b (1-\phi_2(R))^d  }\\
    p_b(b,d) &= \begin{cases}
        \frac{q_b}{q_1(1-\phi_1(R))^d + q_2(1-\phi_2(R))^d + q_b} & \text{if } b=0\\
        0 & \text{if } b>0\\
    \end{cases}
\end{align*}
The formulas for $p_1$ and $p_2$ are stated for $b, d > 0$, where objectively good and bad products are ruled out. When $d = 0$ or $b = 0$, $p_1$ and $p_2$ follow from $p_g + p_1 + p_2 + p_b = 1$ and the ratio $p_1/p_2 = q_1\phi_1^b(1-\phi_1)^d / [q_2\phi_2^b(1-\phi_2)^d]$.

\begin{lemma} \label{lemma:multi_recommendations}
Consider the model with multiple ratings. Suppose Assumption \ref{ass:symmetry} holds and consider any standard $R\in(0,1)$.
\begin{enumerate}
    \item[(i)] If the receiver obtains only likes (dislikes), then, as the number of ratings increases, her posterior belief converges to $p_g=1$ ($p_b=1$).
    \item[(ii)] If the receiver obtains mixed ratings, then $p_1(b, d)/p_2(b,d) = q_1/q_2$ and $p_g(b,d)= p_b(b,d)= 0$ for any $b,d>0$.
\end{enumerate}
\end{lemma}
\begin{proof}
    (i) Suppose the receiver gets only likes and that $b\rightarrow\infty$. Then, $\phi_1(R) = \phi_2(R) =:\phi$. Since $F$ has full support and $R\in(0,1)$, $\phi\in (0,1)$, so that $\lim_{b\rightarrow \infty} (\phi_i(R))^b = 0$ for $i=1, 2$. Hence, $\lim_{b\rightarrow \infty} p_g(b, 0) = 1$. The argument for dislikes is analogous. 

    (ii) For any $b,d>0$, $p_g(b,d)= p_b(b,d)= 0$ obviously holds. Further, because $\phi_1(R) = \phi_2(R)$ by Assumption \ref{ass:symmetry}, $p_i(b, d) = \frac{q_i}{q_1 + q_2}$  for $i=1,2$, yielding statement (ii).
\end{proof}

\subsubsection*{Proof of Proposition \ref{prop:no_value_infinite_learning}}

The proof will follow from Steps 1-4 below. We use the notation $b(i)
:=
\frac{q_1}{q_1+q_2}\Bigl(\frac12+i\Bigr)
+
\frac{q_2}{q_1+q_2}\Bigl(\frac12-i\Bigr)
-
U^0(i)$ for the receiver's gain from buying a controversial product after infinitely many ratings, relative to the outside option.

\textbf{Step 1:} \textit{the value of the infinite-learning system is
\[
V_\infty
=
q_g(1-q_g-Q)+2Q\,\mathbb{E}\bigl[b(i)^+\bigr],
\qquad
b(i)^+:=\max\{b(i),0\}.
\]}

By Lemma~\ref{lemma:multi_recommendations}, infinite learning perfectly reveals objectively good and bad products, while mixed ratings identify a controversial product but not its variant. Hence, (i) good products are always bought, (ii) bad products are never bought, and (iii) controversial products are bought exactly by those types with $b(i)\ge 0$. Thus, the result holds.\footnote{The first term is the gain from correctly identifying objectively good products: a good product yields payoff~1, while the outside option yields $\mathbb{E}[U^0(i)]=q_g+Q$ under symmetric $F$, so the gain is $q_g(1-q_g-Q)$. Objectively bad products contribute zero because they are identified and not bought.} 

\textbf{Step 2:} \textit{The optimally designed single-rating system gives the value }
\begin{equation*}
    V_1^*= q_g(1-q_g-Q)+2Q\,\max\{\mathbb{E}[b(i)],0\}.
\end{equation*}
Under Assumption~\ref{ass:symmetry}, Proposition~\ref{prop:solution_symmetric_population} implies that the optimal single-standard design is extreme:
\begin{itemize}
    \item if $\sigma<1$, the designer chooses $R\to 0$, so all controversial products are bought;
    \item if $\sigma>1$, the designer chooses $R\to 1$, so no controversial product is bought;
    \item if $\sigma=1$, all standards are optimal.
\end{itemize}
Hence the benchmark value is $V_1^*= q_g(1-q_g-Q)+2Q\,\max\{\mathbb{E}[b(i)],0\}.$ 

\textbf{Step 3:} $V_\infty = V_1^*
\quad\Longleftrightarrow\quad
b(i)$ \textit{has a constant sign almost surely}.

 Since $b=b^+-b^-$ with $b^-:=\max\{-b,0\}$, we have $\mathbb{E}[b^+] \ge \max\{\mathbb{E}[b],0\}.$ Moreover, 
\begin{itemize}
    \item if $\mathbb{E}[b]\ge 0$, then $\mathbb{E}[b^+]-\mathbb{E}[b]=\mathbb{E}[b^-]$, so equality requires $b^-=0$ a.s.;
    \item if $\mathbb{E}[b]\le 0$, then $\mathbb{E}[b^+]=0$ requires $b^+=0$ a.s.
\end{itemize}
Thus the claim holds.

\textbf{Step 4:} $b(i)$ \textit{has a constant sign almost surely if and only if one of the following holds:}
    \begin{align*}
        \text{(i)}\ \lambda = 1;\quad
       \text{(ii)}\ \lambda > 1\ \text{and}\ \sigma \notin \left(\frac{1}{\lambda}, \lambda\right);\quad
        \text{(iii)}\ \lambda < 1\ \text{and}\ \sigma \notin \left(\lambda, \frac{1}{\lambda}\right).
    \end{align*}

Using $q_1=\frac{2Q\lambda}{\lambda+1}$, $q_2=\frac{2Q}{\lambda+1}$, and $q_g=\frac{(1-2Q)\sigma}{1+\sigma}$:
\[
b(i)
=
\frac{(1-2Q)(\lambda-1)}{\lambda+1}\,(i-\tilde i_\infty),
\qquad
\tilde i_\infty
=
\frac{\sigma-1}{2(1+\sigma)}\frac{\lambda+1}{\lambda-1}.
\]

\textbf{Case 1: $\lambda=1$.}
Then $b(i)$ is constant in $i$, so it has a constant sign a.s. Hence $V_\infty=V_1^*$.

\textbf{Case 2: $\lambda>1$.}
Then $b(i)$ is strictly increasing in $i$. Because $F$ has full support on $[-1/2,1/2]$, $b(i)$ has a constant sign a.s.\ iff $\tilde i_\infty\le -1/2$ or $\tilde i_\infty\ge 1/2$. The first inequality is equivalent to $\sigma\le 1/\lambda$, and the second to $\sigma\ge\lambda$. Hence $V_\infty=V_1^*\iff\sigma\notin(1/\lambda,\lambda)$.

\textbf{Case 3: $\lambda<1$.}
Then $b(i)$ is strictly decreasing in $i$. By the same argument, $b(i)$ has a constant sign a.s.\ iff $\tilde i_\infty\ge 1/2$ or $\tilde i_\infty\le -1/2$. Now $\lambda-1<0$, so the direction of the inequalities reverses: $\tilde i_\infty\ge 1/2$ gives $\sigma\le\lambda$, and $\tilde i_\infty\le -1/2$ gives $\sigma\ge 1/\lambda$. Hence $V_\infty=V_1^*\iff\sigma\notin(\lambda,1/\lambda)$.

At the boundary values (e.g., $\sigma=1/\lambda$ with $\lambda>1$), $\tilde i_\infty=-1/2$, so $b(-1/2)=0$ and $b(i)>0$ for $i>-1/2$. Under full support with continuous $F$, zero at a single point has measure zero, so $b\ge 0$ a.s. The boundary belongs to the equality region, yielding the open intervals in the statement.


\subsubsection{Distinct Sender and Receiver Distributions}\label{AppDistinct}

We prove the following result:
\begin{proposition} \label{prop:solution_distinct_population}
    Consider the model with distinct preference distributions of senders and receivers. Suppose the sender distribution $F$ satisfies Assumption~\ref{ass:symmetry}. Then, all types are obedient. Further, $V(R) = \pi^L_F[\Delta_{O,F}^L-\Delta_{S,F}^L \E_G[i]]$ is 
    \begin{enumerate}
        \item[(i)] decreasing in $R$ if $\sigma<\frac{Q + (q_1-q_2) \E_G[i]}{Q - (q_1-q_2) \E_G[i]}$;
        \item[(ii)]  increasing in $R$  if $\sigma>\frac{Q + (q_1-q_2) \E_G[i]}{Q - (q_1-q_2) \E_G[i]}$;
        \item[(iii)]  constant in $R$  if $\sigma=\frac{Q + (q_1-q_2) \E_G[i]}{Q - (q_1-q_2) \E_G[i]}$.
    \end{enumerate}
\end{proposition}
\begin{proof}
    
The value of the rating system can be written as
\begin{small}
    \begin{align*}
    V(\mathcal{R}) = \begin{cases}
        \int_{-1/2}^{1/2} (V_{A,F}(i) - V_{0,F}(i))dG(i) & \text{if } \vert\Delta_{S,F}^L \vert\leq 2\Delta_{O,F}^L\\
        \int_{-1/2}^{\tilde{i}}V_{N,F}(i) dG(i)  +\int_{\tilde{i}}^{1/2} V_{A,F}(i) dG(i) - \int_{-1/2}^{1/2} V_{0,F}(i) dG(i)  & \text{if } \Delta_{S,F}^L<-2\Delta_{O,F}^L\\
        \int_{-1/2}^{\tilde{i}}V_{A,F}(i) dG(i) +  \int_{\tilde{i}}^{1/2} V_{N,F}(i) dG(i)  - \int_{-1/2}^{1/2} V_{0,F}(i) dG(i) & \text{if } \Delta_{S,F}^L>2\Delta_{O,F}^L 
    \end{cases}
    \end{align*}
\end{small}
and the result from Proposition \ref{prop:value_of_RS} translates accordingly so that we can rewrite  $V(\mathcal{R})$ as
\begin{footnotesize}
    \begin{align*}
        &\pi^L[\Delta_{O,F}^L-\Delta_{S,F}^L \E_G[i]]  &\text{if } \vert\Delta_{S,F}^L \vert\leq 2\Delta_{O,F}^L\\
        & (1-\pi^L)G(\tilde{i}) \left[\Delta_{O,F}^D  - \Delta_{S,F}^D\E_G[i \mid i\leq \tilde{i}]\right]+ \pi^L(1-G(\tilde{i}))\left[\Delta_{O,F}^L - \Delta_{S,F}^L  \E_G[i \mid i\geq \tilde{i}]\right] & \text{if } \Delta_{S,F}^L<-2\Delta_{O,F}^L\\
     &  \pi^LG(\tilde{i})\left[\Delta_{O,F}^L  - \Delta_{S,F}^L \E_G[i \mid i\leq \tilde{i}]\right]+ (1-\pi^L)(1- G(\tilde{i})) \left[\Delta_{O,F}^D - \Delta_{S,F}^D \E_G[i \mid i\geq \tilde{i}]\right]&\text{if } \Delta_{S,F}^L>2\Delta_{O,F}^L.
\end{align*}       
\end{footnotesize}
   Since $F$ satisfies Assumption \ref{ass:symmetry}, all receivers will be obedient so that we are in the case $\vert\Delta_{S,F}^L \vert\leq 2\Delta_{O,F}^L$. Thus, we have 
    \begin{align*}
        V(\beta) &= \pi^L[\Delta_{O,F}^L-\Delta_{S,F}^L \E_G[i]] \\
        &= q_g + \beta Q - (q_g + 2 \beta Q)(q_g + Q) + (q_2-q_1)(q_g -\beta( 1-2Q))\E_G[i]
    \end{align*}
    and therefore  $V'(\beta) = Q - 2  Q(q_g + Q) - (q_2-q_1)(1-2Q)\E_G[i]$. Replacing  $q_g = ((1-2Q)\sigma)/(\sigma+1)$ and rearranging the inequalities $V'(\beta)>0$ and $V'(\beta)<0$ yields the statement in the proposition. 
\end{proof}

\subsubsection{Polarization and the Prevalence of Controversial Products}\label{App:comparative_statics}

\paragraph{Polarization.}

Under Assumption~\ref{ass:symmetry}, polarization affects the rating system through $\beta = F(1/2-R),$ the share of senders who recommend a controversial product at rating standard $R$. By Lemma~\ref{lemma:solution_symmetric_population}, the value of the rating system is strictly increasing in $\beta$ when $\sigma<1$, strictly decreasing in $\beta$ when $\sigma>1$, and independent of $\beta$ when $\sigma=1$.

Figure~\ref{fig:mean-spread} illustrates such a change in $\beta$. The intuition is that a higher $\beta$ makes likes more frequent but less objectively informative. To see this, suppose $\sigma<1$ and $R>1/2$. As shown in Figure~\ref{fig:senders}, when $R>1/2$, only sufficiently extreme types issue likes for controversial products. Hence a higher $\beta$ means that likes are more likely to stem from controversial rather than objectively good products. This raises the probability of likes while lowering their objective value. When good products are rare, the frequency effect dominates; when good products are common, the informativeness effect dominates.

    \begin{figure}
        \centering
        \caption{Population polarization}
            \label{fig:mean-spread}
\begin{tikzpicture}[scale=3, x=2cm, y=1cm]
    \draw[->] (-0.6,0) -- (0.6,0) node[right] {$i$};
    \draw[->] (0,0) -- (0,1.2);
    
    \draw[dashed] (0,0) -- (0,1);
    
    \draw (-0.5,0) -- (0.5,0);
    \foreach \x in {-0.5,-0.25,0,0.25,0.5}
        \draw (\x,0.05) -- (\x,-0.05) node[below, font=\footnotesize] {\x};
    
    \draw[dotted, thick, black] (-0.5,0) -- (-0.25,0.5);
    \draw[dotted,thick, black] (-0.25,0.5) -- (0.25,0.5);
    \draw[dotted,thick, black] (0.25,0.5) -- node[below right, black, font=\footnotesize] {$\beta=\frac{1}{2}$} (0.5,1);
    
    \draw[dashed,thick, black] (-0.5,0) -- (-0.25,0);
    \draw[dashed,thick, black] (-0.25,0) -- (0.25,1);
    \draw[dashed,thick, black] (0.25,1) -- node[above, black, font=\footnotesize] {$\beta=0$} (0.5,1);
    
    \node[left, font=\footnotesize] at (0,1) {$1$};
    \node[left, font=\footnotesize] at (0,0.5) {$\frac{1}{2}$};
\end{tikzpicture}

        \flushleft \footnotesize \textit{Notes:} We have fixed $R=3/4$. The dotted and dashed lines both correspond to symmetric type distributions $F$, where $\beta=1/2$ for the dotted line and $\beta=0$ for the dashed line. A move from $\beta=0$ to $\beta=1/2$ illustrates a symmetric shift in the type distribution that increases the share of types willing to recommend a controversial product at the fixed standard $R=3/4$.
    \end{figure}
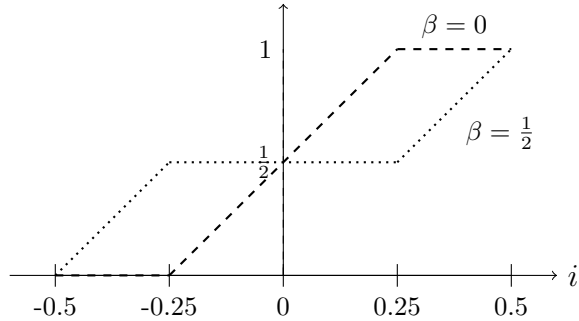

The impact of preference heterogeneity on sales is straightforward. Under Assumption~\ref{ass:symmetry}, all types are obedient so sales equal the probability of a like and are therefore increasing in $\beta$. Thus, shifts that increase $\beta$ also increase sales.

\paragraph{Prevalence of Controversial Products.}

The following lemma gives the comparative statics of $V$ in $Q$:
\begin{lemma}\label{lemma:comparative_statics_Q}
    Suppose Assumption \ref{ass:symmetry}~holds.
    \begin{itemize}
        \item[(i)] If $3\sigma-\beta-\sigma^2+\sigma^2\beta>0$, then the value of the rating system is decreasing in $Q$. In particular, this is the case when $\sigma=1$.
        \item[(ii)] If $3\sigma-\beta-\sigma^2+\sigma^2\beta<0$, then $Q^*:=\frac{3\sigma-\beta-\sigma^2+\sigma^2\beta}{4\sigma-4\beta-4\sigma^2+4\sigma^2\beta}\in (0,\frac{1}{2})$ maximizes $V(R)$.
    \end{itemize}
\end{lemma}

\begin{proof}

We first take the derivative of $V$ with respect to $Q$ to obtain
\begin{gather}
\frac{\partial V}{\partial Q}=
\frac{1}{\left(  \sigma+1\right)  ^{2}}\left(  \beta-3\sigma+4Q\sigma
-4Q\beta+\sigma^{2}-4Q\sigma^{2}-\sigma^{2}\beta+4Q\sigma^{2}\beta\right)
\nonumber
\end{gather}

Grouping the $Q$-terms, this can be written in affine form as
\begin{align*}
    \frac{\partial V}{\partial Q} = \frac{-(3\sigma - \beta - \sigma^2 + \sigma^2\beta) + 4Q(\sigma - \beta - \sigma^2 + \sigma^2\beta)}{(\sigma+1)^2},
\end{align*}
which is linear in $Q$.

\smallskip
\noindent\textit{Proof of (i).} Suppose $3\sigma - \beta - \sigma^2 + \sigma^2\beta > 0$. Straightforward calculations show that $\partial V/\partial Q$ is negative for $Q=0$ and negative for $Q=1/2$.
Because it is affine in $Q$, it is thus negative throughout $(0,1/2)$. Hence $V$ is decreasing in $Q$, proving (i). In particular, at $\sigma=1$, $3\sigma-\beta-\sigma^2+\sigma^2\beta = 2$ for every $\beta\in[0,1]$, so (i) applies.

\smallskip
\noindent\textit{Proof of (ii).} Suppose $3\sigma-\beta-\sigma^2+\sigma^2\beta < 0$.
Setting $\partial V/\partial Q$ equal to zero and solving for the candidate solution we obtain
\[
Q^{\ast}=\frac{3\sigma-\beta-\sigma^{2}+\sigma^{2}\beta}{4\sigma
-4\beta-4\sigma^{2}+4\sigma^{2}\beta}%
\]
Next, the SOC reads $4\sigma-4\beta-4\sigma^{2}+4\sigma^{2}\beta<0,$ where the l.h.s. of the SOC is the denominator of $Q^{\ast}$. Using the identity $4\sigma-4\beta-4\sigma^{2}+4\sigma^{2}\beta = 4(3\sigma-\beta-\sigma^{2}+\sigma^{2}\beta) - 8\sigma$, the condition $3\sigma-\beta-\sigma^{2}+\sigma^{2}\beta<0$ (together with $\sigma>0$) implies the SOC holds. Hence, the candidate solution is positive. It remains to show that $Q^{\ast}<\frac{1}{2}$. The requirement $\frac{3\sigma-\beta-\sigma^{2}+\sigma^{2}\beta}{4\sigma-4\beta-4\sigma^{2}+4\sigma^{2}\beta}<1/2$ becomes
\[
2\sigma-2\beta-2\sigma^{2}+2\sigma^{2}\beta<3\sigma-\beta-\sigma^{2}%
+\sigma^{2}\beta,
\]
because the denominator of the ratio is negative. Equivalently, $\left(  \sigma+1\right)  \left(  \sigma(1-\beta)+\beta\right)  >0,$ which always holds for $\beta \leq 1$.
\end{proof}

As an immediate consequence we obtain:
\begin{corollary}\label{cor:comparative_statics_Q}Suppose Assumption \ref{ass:symmetry}~holds. Then for every $\beta \in (0,1)$, there exist $\underline{\sigma}(\beta) \in (0,1)$ and $\overline{\sigma}(\beta)>1$ such that 
    \begin{itemize}
        \item[(i)] the value of the rating system is decreasing in $Q$ on  $(\underline{\sigma}(\beta),\overline{\sigma}(\beta))$.
        \item[(ii)] For $\sigma<\underline{\sigma} $ and  $\sigma>\overline{\sigma} $, the value of the rating system is maximized for some $Q^*>0 $.

    \end{itemize}
\end{corollary}
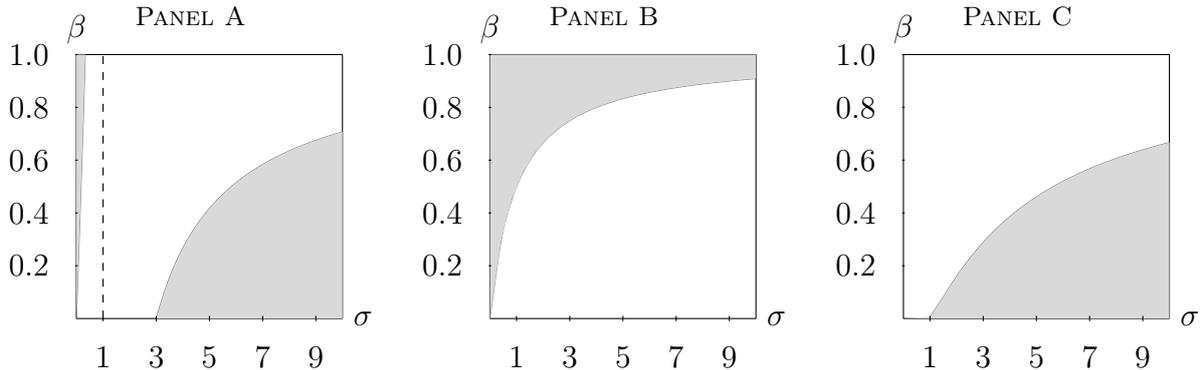
\begin{figure}[htbp]
    \centering
    \caption{Illustration of Corollary~\ref{cor:comparative_statics_Q}}
        \label{fig:prevalence_controversial}
    \begin{subfigure}[b]{0.32\textwidth}
    \captionsetup{skip=-0.4cm}
      \subcaption{Panel A}      
        \centering
\begin{tikzpicture}[x=0.35cm,y=3.5cm]
    \draw[-] (0,0) -- (10,0) node[right] {$\sigma$};
    \draw[-] (0,0) -- (0,1) node[above] {$\beta$};
    \draw[-] (10,0) -- (10,1);
    \draw[-] (10,1) -- (0,1);
    
    \draw[domain=0:1/3,smooth,variable=\s,gray] plot ({\s},{(-3*\s + \s^2)/(-1 + \s^2)});
    
    \draw[domain=3:10,smooth,variable=\s,gray] plot ({\s},{(-3*\s + \s^2)/(-1 + \s^2)});
    
    \fill[gray!30, domain=0:1/3, variable=\s]
        plot ({\s},{(-3*\s + \s^2)/(-1 + \s^2)})
        -- (0,1) -- (0,0) --cycle;
    
    \fill[gray!30, domain=3:10, variable=\s]
        plot ({\s},{(-3*\s + \s^2)/(-1 + \s^2)})
        -- (10,0) -- cycle;
    
    \draw[dashed] (1,0) -- (1,1);
    
    \foreach \x in {1,3,5,7,9} {
        \draw (\x,0.01) -- (\x,-0.01) node[below=0.2cm] {$\x$};
    }
    \foreach \y in {0.2,0.4,0.6,0.8,1.0} {
        \draw (0.05,\y) -- (-0.05,\y) node[left=0.2cm] {$\y$};
    }
\end{tikzpicture}

    \end{subfigure}
    \hfill
    \begin{subfigure}[b]{0.32\textwidth}
    \captionsetup{skip=-0.4cm}
      \subcaption{Panel B}

        \centering
        \begin{tikzpicture}[x=0.35cm,y=3.5cm]
            \draw[-] (0,0) -- (10,0) node[right] {$\sigma$};
            \draw[-] (0,0) -- (0,1) node[above] {$\beta$};
            \draw[-] (10,0) -- (10,1);
            \draw[-] (10,1) -- (0,1);
            
            \draw[domain=0:10,smooth,variable=\s,gray] plot ({\s},{\s/(1 + \s)});
            
            \fill[gray!30, domain=0:10, variable=\s]
                plot ({\s},{\s/(1 + \s)})
                -- (10,1) -- (0,1) -- (0,0) --cycle;
            
            \foreach \x in {1,3,5,7,9} {
                \draw (\x,0.01) -- (\x,-0.01) node[below=0.2cm] {$\x$};
            }
            \foreach \y in {0.2,0.4,0.6,0.8,1.0} {
                \draw (0.05,\y) -- (-0.05,\y) node[left=0.2cm] {$\y$};
            }
        \end{tikzpicture}
        
    \end{subfigure}
    \hfill
    \begin{subfigure}[b]{0.32\textwidth}
    \captionsetup{skip=-0.4cm}
     \subcaption{Panel C}
        \centering
        \begin{tikzpicture}[x=0.35cm,y=3.5cm]
            \draw[-] (0,0) -- (10,0) node[right] {$\sigma$};
            \draw[-] (0,0) -- (0,1) node[above] {$\beta$};
            \draw[-] (10,0) -- (10,1);
            \draw[-] (10,1) -- (0,1);
            
            \draw[domain=1.01:10,smooth,variable=\s,gray] plot ({\s},{(\s*(16.5 + 8.5*\s - 12.5*sqrt(1.64 + 0.36*\s)*sqrt(1 + \s)))/(-1 + \s^2)});
            
            \fill[gray!30, domain=1.01:10, variable=\s]
                plot ({\s},{(\s*(16.5 + 8.5*\s - 12.5*sqrt(1.64 + 0.36*\s)*sqrt(1 + \s)))/(-1 + \s^2)})
                -- (10,0) -- (0,0) --cycle;
            
            \foreach \x in {1,3,5,7,9} {
                \draw (\x,0.01) -- (\x,-0.01) node[below=0.2cm] {$\x$};
            }
            \foreach \y in {0.2,0.4,0.6,0.8,1.0} {
                \draw (0.05,\y) -- (-0.05,\y) node[left=0.2cm] {$\y$};
            }
        \end{tikzpicture}
       
    \end{subfigure}
   \flushleft\footnotesize \textit{Notes:} In Panel~A, the shaded area indicates the parameter space for which an interior solution for $Q$ arises by Corollary \ref{cor:comparative_statics_Q}. In Panel~B, the shaded area indicates where $\pi^L$ is increasing in $Q$. In Panel~C, the shaded area indicates where $\Delta_O^L$ is increasing in $Q$, illustrated for $Q=0.1$.
\end{figure}
Figure~\ref{fig:prevalence_controversial} illustrates the result. In the white area in Panel~A, $V(Q)$ is decreasing in $Q$; in the shaded area, an interior maximizer exists. When objectively good and bad products are similarly likely ($\sigma$ is close to $1$), a higher prevalence of controversial products indeed reduces $V(Q)$. For very high or low values of $\sigma$, however, this need not be true. This surprising result reflects the multiplicative structure of $V(Q)$. Under Assumption~\ref{ass:symmetry}, $V(Q)$ is the product of the buying probability $\pi^L$ and the objective value of the rating per buyer, $\Delta_O^L$. The effect of $Q$ on $\pi^L$ is negative only if $\sigma$ is sufficiently large (see Panel~B). Intuitively, as $Q$ increases, objectively good products are less likely, reducing the probability of likes. However, there are now more likes stemming from controversial products, as these have become more likely. $\beta$ and $\sigma$ determine which of these effects dominates. When $\beta$ is high, the greater likelihood that controversial products receive likes increases the total probability of likes, strengthening the positive effect. When $\sigma$ is high, the lower share of objectively good products dominates, strengthening the negative effect. 
Similarly, one can show that the effect of $Q$ on $\Delta_O^L$ becomes negative only for small $\sigma$ (see Panel C), reflecting the increased probability that a product with a like is controversial.
\footnote{There is also a positive effect of higher $Q$, reflecting a change in the value of the outside option. The values of $\sigma$ and $\beta$ determine which effect dominates.}

\newpage
\bibliographystyle{ecta}
\bibliography{references_reco.bib}

\clearpage
\setcounter{page}{1}
\setcounter{section}{0}
\setcounter{subsection}{0}
\setcounter{subsubsection}{0}
\setcounter{equation}{0}
\setcounter{figure}{0}
\setcounter{table}{0}
\setcounter{footnote}{0}
\setcounter{theorem}{0}
\setcounter{proposition}{0}
\setcounter{corollary}{0}
\setcounter{lemma}{0}
\renewcommand{\thesection}{OA}
\renewcommand{\thesubsection}{OA.\arabic{subsection}}
\renewcommand{\thesubsubsection}{\thesubsection.\arabic{subsubsection}}
\renewcommand{\theequation}{OA.\arabic{equation}}
\renewcommand{\thefigure}{OA.\arabic{figure}}
\renewcommand{\thetable}{OA.\arabic{table}}
\renewcommand{\theproposition}{OA.\arabic{proposition}}
\renewcommand{\thecorollary}{OA.\arabic{corollary}}
\renewcommand{\thelemma}{OA.\arabic{lemma}}
\section{Online Appendix}\label{Sec:OnlineAppendix}
\subsection{Value of Rating Systems}
\subsubsection{Universal Obedience in Example~\ref{ass:asymmetry}}\label{App:universal_acceptance_cutoff}

Recall that, for the economic analysis, we maintain the assumptions that $R\in(0,1)$ and $q_s>0$ for $s $, implying $Q<1/2$. For continuity and convexity arguments, we occasionally extend functions of $R$ and $Q$ to the closures $[0,1]$ and $[0,1/2]$.

The following two lemmas supply the exact cutoff $\bar Q(a,\sigma)$ invoked informally before Proposition~\ref{prop:benchmark_characterization}, clarifying under which conditions the requirement that all receivers are obedient is satisfied.\footnote{Whenever we write $\bar Q(a,\sigma)=1/2$, this should be understood as a boundary statement: universal acceptance then holds for every admissible $Q<1/2$.} Lemma~\ref{lemma:universal_acceptance_cutoff} establishes existence, uniqueness, and an explicit characterization for $a>1$. Lemma~\ref{lemma:cutoff_monotonicity_a} describes how the cutoff varies with $a$.

\begin{lemma}\label{lemma:universal_acceptance_cutoff}
Suppose we are in Example~\ref{ass:asymmetry}. Then, for every $(a,\sigma) \in \mathcal{R}_+ \times \mathcal{R}_+$, there exists a unique cutoff $\bar Q(a,\sigma)\in(0,1/2]$ such that receivers are universally obedient for $(Q,a,\sigma)$ if and only if $Q\leq \bar Q(a,\sigma)$. Moreover:
\begin{enumerate}
    \item[(i)] For $a=1$ and every $(Q,\sigma) \in  (0,1/2) \times \mathcal{R}_+ $, receivers are universally obedient for $(Q,a,\sigma)$.
    \item[(ii)] For $a>1$ and $(Q,\sigma) \in  (0,1/2) \times \mathcal{R}_+ $, receivers are universally obedient for $(Q,a,\sigma)$ if and only if
    \[
        N(R;Q,a,\sigma):= q_gq_b-Q^2+Q(q_b+Q)(1-R)^a+Q(q_g+Q)R^a \geq 0
    \]
    for all $R\in(0,1)$.
    \item[(iii)]  For $a<1$, receivers are universally obedient if and only if
    \[
        \tilde N(R;Q,a,\sigma):= q_gq_b-Q^2+Q(q_b+Q)(1-R^a)+Q(q_g+Q)\bigl(1-(1-R)^a\bigr) \geq 0
    \]
    for all $R\in(0,1)$.
\end{enumerate}
\end{lemma}

\begin{proof}
In Steps 1-7, we establish the lemma for $a\geq 1$. Steps 1 and 2 hold for all $a$. In Example~\ref{ass:asymmetry},  $F(i)=(i+1/2)^a$ and $q_1=q_2=Q$, so $\phi_1(R)=1-R^a$ and $\phi_2(R)=(1-R)^a$. To simplify notation, we set $s_a(R):=\phi_1(R)+\phi_2(R)$ and $b_a(R):=\phi_2(R)-\phi_1(R)=R^a+(1-R)^a-1$, 

\textbf{Step 1.} \textit{We have}
\begin{equation*}
    \Delta_S^L(R)\ =\ \frac{Q\,b_a(R)}{\pi^L(R)},\qquad
    \Delta_O^L(R)\ =\ \frac{q_g+\tfrac{Q}{2}s_a(R)}{\pi^L(R)}-(q_g+Q).
\end{equation*}
\textbf{Proof of Step 1: }In Example~\ref{ass:asymmetry}, the probability of like is $\pi^L(R)=q_g+Qs_a(R)>0$. The posteriors after a like are $P_g^L=q_g/\pi^L$, $p_1^L=Q\phi_1/\pi^L$, and $p_2^L=Q\phi_2/\pi^L$, so Step 1 follows from by Definition~\ref{def:effects}.
\medskip

\textbf{Step 2.}
(i) \textit{For} $a>1$, $\Delta_S^L(R)<0$. (ii) \textit{For} $a=1$,   $\Delta_S^L\equiv 0$,\textit{ so that all types are universally obedient}. (iii) \textit{For } $a<1$, $\Delta_S^L(R)>0$.

\textbf{Proof of Step 2}: For $a>1$, strict convexity of $x\mapsto x^a$ on $[0,1]$ gives $b_a(R)<0$ on $(0,1)$, hence $\Delta_S^L(R)<0$. For $a=1$, $b_a\equiv 0$, so $\Delta_S^L\equiv 0$ and Lemma~\ref{prop:pop_symmetry_allfollow} implies that receivers are universally obedient; in particular $\bar Q(1,\sigma)=1/2$, proving~(ii). For $a<1$, strict concavity of $x\mapsto x^a$ gives $b_a(R)>0$, hence $\Delta_S^L(R)>0$.

\medskip
\textbf{Step 3.} \textit{For $a>1$, the objective rating  effect $\Delta_O^L(R)>0$ is strictly positive}\textit{ for all }$R\in[0,1]$.

\textbf{Proof of Step 3:} Multiplying $\Delta_O^L(R)$ by $\pi^L(R)$, substituting $\pi^L=q_g+Qs_a$ and rearranging yields $\pi^L(R)\Delta_O^L(R)\ =\ q_gq_b+\frac{Q}{2}\bigl(s_a(R)q_b+(2-s_a(R))q_g\bigr).$
Since $0\leq s_a(R)\leq 2$ and $q_g,q_b>0$ for $Q<1/2$, both summands are nonnegative and the first is strictly positive. Hence $\Delta_O^L(R)>0$ for all $R\in(0,1)$, and by continuity the same inequality extends to the boundary points $R=0$ and $R=1$.

\medskip
\textbf{Step 4.}
\textit{For} $(Q,a,\sigma)$ with $a>1$, \textit{receivers are universally obedient if and only if} $N(R;Q,a,\sigma)\geq 0$ \textit{for all} $R\in(0,1)$.

\textbf{Proof of Step 4:} By Corollary~\ref{cor:follow}, all receiver types are obedient for the standard $R$ if and only if $|\Delta_S^L(R)|\leq 2\Delta_O^L(R)$. Since $\Delta_S^L<0$ for $a>1$, this is equivalent to $\Delta_O^L(R)+\tfrac12\Delta_S^L(R)\geq 0$. Multiplying by $\pi^L(R)>0$ and substituting the closed forms from Step~1 gives
\begin{align*}
    \pi^L(R)\bigl(\Delta_O^L(R)+\tfrac12\Delta_S^L(R)\bigr)\ =\ q_g+\frac{Q}{2}\bigl(s_a(R)+b_a(R)\bigr)-(q_g+Q)\pi^L(R) \geq 0
\end{align*}
Using $s_a(R)+b_a(R)=2(1-R)^a$ and $\pi^L(R)=q_g+Qs_a(R)=q_g+Q\bigl(1-R^a+(1-R)^a\bigr)$, this rearranges to $\pi^L(R)\bigl(\Delta_O^L(R)+\tfrac12\Delta_S^L(R)\bigr)\ =\ N(R;Q,a,\sigma) \geq 0.$
Therefore, for $(Q,a,\sigma)$ receivers are universally obedient if and only if $N(R;Q,a,\sigma)\geq 0$ for all $R\in(0,1)$. Because $N$ extends continuously to $[0,1]$, this is equivalent to $N(R;Q,a,\sigma)\geq 0$ for all $R\in[0,1]$.

\medskip
\textbf{Step 5.}
\textit{For all} $(a,\sigma) \in (1,\infty) \times \mathcal{R}_+$, \textit{there exists }$\varepsilon(a,\sigma)>0$ \textit{such that all receivers are obedient for all} $(a,\sigma,Q)$ \textit{such that} $Q \leq \varepsilon(a,\sigma)$  \textit{and all} $R\in (0,1)$.

\textbf{Proof of Step 5:} Since $|b_a(R)|$ is continuous on the compact interval $[0,1]$, $C_a:=\max_{R\in[0,1]}|b_a(R)|<\infty$. Using $\pi^L(R)\geq q_g$ from Step~1, $|\Delta_S^L(R)|\leq QC_a/q_g$ for every $R$. From Step~3, $\pi^L\Delta_O^L\geq q_gq_b$, and $\pi^L(R)\leq q_g+2Q$ since $s_a\leq 2$, so $\Delta_O^L(R)\geq q_gq_b/(q_g+2Q)$. As $Q\downarrow 0$,
\begin{align*}
    \frac{QC_a}{q_g}\ \to\ 0,\qquad \frac{2q_gq_b}{q_g+2Q}\ \to\ \frac{2}{1+\sigma}\ >\ 0.
\end{align*}
Hence there exists $\varepsilon(a,\sigma)>0$ such that $|\Delta_S^L(R)|<2\Delta_O^L(R)$ uniformly in $R\in[0,1]$ for every $Q\leq\varepsilon(a,\sigma)$. By Corollary \ref{cor:follow}, receivers are universally obedient for every such $Q,a,\sigma$.

\medskip
\textbf{Step 6.}
\textit{Define}
$\mathcal{S}(a,\sigma)\ :=\ \bigl\{Q\in[0,1/2)\ \mid \text{all $i$ are universally obedient at } Q(a,\sigma)\bigr\}$ \textit{as the obedience set. Then, for $a>1$, either}
 $\mathcal{S}(a,\sigma)=[0,\bar Q(a,\sigma)]$ \textit{for some }$\bar Q(a,\sigma) \in (0,1/2)$ \textit{or} $\mathcal{S}(a,\sigma)=[0,1/2)$.

\textbf{Proof of Step 6}

\emph{$\mathcal{S}(a,\sigma)$ is nonempty:} More generally, by Step~5, $[0,\varepsilon(a,\sigma)]\subseteq\mathcal{S}(a,\sigma)$.

\emph{$\mathcal{S}(a,\sigma)$ is downward closed.} By Corollary \ref{cor:follow}, all receivers are obedient for the standard $R$ if and only if $\tilde i(R,Q)\leq -1/2$. Writing $\delta:=(\sigma-1)/(\sigma+1)\in(-1,1)$, direct differentiation at fixed $\sigma$ (taking the dependence of $q_g=(1-2Q)\sigma/(1+\sigma)$ and $q_b=(1-2Q)/(1+\sigma)$ on $Q$ into account) yields
\begin{align}\label{eq:i_tilde}
    \frac{\partial \tilde i(R,Q)}{\partial Q}\;=\;-\,\frac{Q^2\,\delta\!\left(R^a-(1-R)^a+\delta\right)+\dfrac{1-\delta^2}{4}}{Q^2\,b_a(R)}.
\end{align}
Letting $x:=R^a-(1-R)^a\in[-1,1]$, the bound $\delta(x+\delta)\geq -|\delta|(1-|\delta|)$ gives
\[
Q^2\,\delta(x+\delta)+\frac{1-\delta^2}{4}\;\geq\;(1-|\delta|)\!\left(\frac{1+|\delta|}{4}-Q^2|\delta|\right)\;>\;0
\]
for $Q<1/2$, so that the numerator in $\eqref{eq:i_tilde}$ is strictly positive. For $a>1$, Step~2 gives $b_a(R)<0$ on $(0,1)$, so $\partial\tilde i(R,Q)/\partial Q>0$: $\tilde i(R,Q)$ is strictly increasing in $Q$. Hence, if $\tilde i(R,Q)\leq -1/2$ at some $Q$, the same inequality holds at every $Q'\leq Q$; equivalently, $\mathcal{S}(a,\sigma)$ is downward closed.

\emph{ $\mathcal{S}(a,\sigma)$ \textit{is relatively closed }in $[0,1/2)$}: Define $G(R,Q):=2\Delta_O^L(R,Q)-|\Delta_S^L(R,Q)|$. As $G$ is continuous on $[0,1]\times[0,1/2)$, $m(Q):=\min_{R\in[0,1]}G(R,Q)$ exists for every $q \in[0,1/2)$ by compactness and is continuous in $Q$. Hence $\mathcal{S}(a,\sigma)=\{Q\in[0,1/2):m(Q)\geq 0\}$ is relatively closed in $[0,1/2)$.

\emph{Conclusion.} Combining non-emptiness, downward closedness and relative closedness, Step 6 follows.
In both cases, $\bar Q(a,\sigma):=\sup\mathcal{S}(a,\sigma)\in(0,1/2]$.

\medskip
\textbf{Step 7.} $\bar Q(a,\sigma)<1/2$ \textit{for} $a>1$.

\textbf{Proof of Step 7:}
To rule out the possibility that the cutoff reaches the boundary $1/2$, we evaluate the continuous extension of $N$ at the boundary point $Q=1/2$. 
At $Q=1/2$, $q_g=q_b=0$, so $N(R;1/2,a,\sigma)=-\tfrac14+\tfrac14\bigl(R^a+(1-R)^a\bigr)$. By strict convexity of $x\mapsto x^a$ for $a>1$, $R^a+(1-R)^a<1$ on $(0,1)$, hence $  N(R;1/2,a,\sigma) <0$  for all $R\in(0,1).$
Now fix any $R_0\in(0,1)$. Since $N$ is continuous in $(R,Q)$, there exists $\delta>0$ such that $N(R_0;Q,a,\sigma)<0$ for every $Q\in(1/2-\delta,1/2)$. By Step~4, universal obedience fails at every such $Q$, so $\bar Q(a,\sigma)\leq 1/2-\delta<1/2$.

Putting together Steps 1-7 implies the result for $a>1$.

\textbf{Step 8.} The proposition also holds for $a<1$.

\textbf{Sketch of proof:} The case $a<1$ is analogous: The main difference is that the sign of $\Delta_S^L$ changes by Step 2. Therefore, the obedience condition $|\Delta_S^L(R)|\leq 2\Delta_O^L(R)$ becomes $\Delta_O^L-\tfrac12\Delta_S^L \geq 0$, and the resulting polynomial is again strictly convex on $[0,1]$ with strictly positive boundary values, so the same chain of arguments delivers a unique cutoff $\bar Q(a,\sigma)<1/2$.
\end{proof}

The universal obedience conditions in Lemma \ref{lemma:universal_acceptance_cutoff} are formulated as requirements that, for all $R \in (0,1)$,  $N( R;Q,a,\sigma) \geq 0$ and $\tilde N(R;Q,a,\sigma) \geq 0$ must hold, respectively. We now show that it is sufficient that the requirements hold at specific values of $R$.  

\begin{corollary}\label{cor:universal_acceptance_cutoff}
    Suppose Example~\ref{ass:asymmetry} holds.
\begin{enumerate}
    \item[(i)] For $a>1$ and $(Q,\sigma) \in  (0,1/2) \times \mathcal{R}_+ $, receivers are universally obedient for $(Q,a,\sigma)$ if and only if $N( R_{\min}(Q,a,\sigma);Q,a,\sigma) \geq 0$, where
    \begin{align*}
        R_{\min}(Q,a,\sigma)\ :=\ \frac{\big((q_b+Q)/(q_g+Q)\big)^{1/(a-1)}}{1+\big((q_b+Q)/(q_g+Q)\big)^{1/(a-1)}}.
    \end{align*}
    \item[(ii)] For $a<1$ and $(Q,\sigma) \in  (0,1/2) \times \mathcal{R}_+ $, receivers are universally obedient for $(Q,a,\sigma)$ if and only if $\tilde N( \tilde R_{\min}(Q,a,\sigma);Q,a,\sigma) \geq 0$, where
 \begin{align*}
        \tilde R_{\min}(Q,a,\sigma)\ :=\ \frac{1}{1+\big((q_b+Q)/(q_g+Q)\big)^{1/(a-1)}}.
    \end{align*}
\end{enumerate}
\end{corollary}
\begin{proof}
(i) Simple calculations show that
  \begin{align*}
    N''(R)\ =\ a(a-1)Q\bigl[(q_b+Q)(1-R)^{a-2}+(q_g+Q)R^{a-2}\bigr]\ >\ 0\quad\text{for }R\in(0,1).
\end{align*}
Thus, $N(\,\cdot\,;Q,a,\sigma)$ is strictly convex on $[0,1]$. Direct evaluation gives $ N(0;Q,a,\sigma)\ =\ q_b(q_g+Q)\ >\ 0$ and $N(1;Q,a,\sigma)\ =\ q_g(q_b+Q)\ >\ 0.$ Setting $N'(R)=0$ yields $(q_g+Q)R_{\min}^{a-1}=(q_b+Q)(1-R_{\min})^{a-1}$, so $R_{\min}$ is the expression in the statement. Since $N$ is strictly convex with positive boundary values, $N\geq 0$ on $[0,1]$ if and only if $N(R_{\min};Q,a,\sigma)\geq 0$. The result thus follows from Lemma \ref{lemma:universal_acceptance_cutoff}.

(ii) For $a<1$, direct differentiation gives
\begin{align*}
    \tilde N''(R)\ =\ aQ(1-a)\bigl[(q_b+Q)R^{a-2}+(q_g+Q)(1-R)^{a-2}\bigr]\ >\ 0\quad\text{for }R\in(0,1),
\end{align*}
so $\tilde N(\,\cdot\,;Q,a,\sigma)$ is strictly convex on $[0,1]$. Direct evaluation gives $\tilde N(0;Q,a,\sigma)=q_b(q_g+Q)>0$ and $\tilde N(1;Q,a,\sigma)=q_g(q_b+Q)>0$. Setting $\tilde N'(R)=0$ yields $(q_g+Q)(1-\tilde R_{\min})^{a-1}=(q_b+Q)\tilde R_{\min}^{a-1}$, so $\tilde R_{\min}$ is the expression in the statement. Strict convexity with positive boundary values again implies $\tilde N\geq 0$ on $[0,1]$ if and only if $\tilde N(\tilde R_{\min};Q,a,\sigma)\geq 0$. The result thus follows from Lemma~\ref{lemma:universal_acceptance_cutoff}. 
\end{proof}

A particularly simple case arises when $\sigma=1$.

\begin{corollary}\label{Cor_Obedience_Ex}
    For $\sigma=1$, receivers are universally obedient at $(Q,a,\sigma)$ if and only if $Q\leq \bar Q(a,1)$, where
    \begin{align*}
        \bar Q(a,1)\ =\ \begin{cases} 2^{a-2}, & 0<a<1,\\ 1/2, & a=1,\\ 1/[4(1-2^{-a})], & a>1.\end{cases}
    \end{align*}
\end{corollary}
\begin{proof}
The result follows from Corollary~\ref{cor:universal_acceptance_cutoff} after noting that at $\sigma=1$ we have $R_{\min}=\tilde R_{\min}=1/2$.
\end{proof}

To illustrate the result, note that $ \bar Q(2,1)=1/3$, so that, if $a=2$, obedience arises for quite high prevalence of controversial products.

\begin{lemma}\label{lemma:cutoff_monotonicity_a}
In Example~\ref{ass:asymmetry}, the cutoff $\bar Q(a,\sigma)$ from Lemma \ref{lemma:universal_acceptance_cutoff} is strictly increasing in $a$ on $(0,1)$, attains its maximum $\bar Q(1,\sigma)=1/2$ at $a=1$, and is strictly decreasing in $a$ on $(1,\infty)$.
\end{lemma}

\begin{proof}
The value at $a=1$ is Lemma~\ref{lemma:universal_acceptance_cutoff}(i). We first prove strict monotonicity on $(1,\infty)$; we then give the corresponding argument on $(0,1)$ using the polynomial $\tilde N$ from Lemma~\ref{lemma:universal_acceptance_cutoff}(iii). Throughout, we use the proof of Step 1 in Lemma~\ref{lemma:universal_acceptance_cutoff} that, for each $a>1$, the acceptance set $\mathcal{S}(a,\sigma)\subset[0,1/2)$ is nonempty, downward closed, and relatively closed, and that $\bar Q(a,\sigma)<1/2$.

\medskip
\textbf{Step 1:} $N(R;Q,a,\sigma)$ \textit{is strictly decreasing in }$a$ \textit{for fixed }$R\in(0,1)$.

\textbf{Proof of Step 1:} By Lemma~\ref{lemma:universal_acceptance_cutoff}(ii), for $a>1$ receivers are universally obedient at $(Q,a,\sigma)$ if and only if $N(R;Q,a,\sigma)\geq 0$ for all $R\in(0,1)$; by continuity this is equivalent to the same condition on $[0,1]$. Since $q_g$ and $q_b$ depend only on $(Q,\sigma)$, only the terms $R^a$ and $(1-R)^a$ in $N$ depend on $a$. Hence
\begin{align*}
    \frac{\partial N(R;Q,a,\sigma)}{\partial a}\ =\ Q(q_b+Q)(1-R)^a\ln(1-R)+Q(q_g+Q)R^a\ln R\ <\ 0\quad\text{for }R\in(0,1),
\end{align*}
because the prefactors are strictly positive while $\ln R<0$ and $\ln(1-R)<0$.

\medskip
\textbf{Step 2: } \textit{Suppose} $a_2>a_1>1$. \textit{Then} $\bar Q(a_2,\sigma)\leq\bar Q(a_1,\sigma)$.

\textbf{ Proof of Step 2:} Take $a_2>a_1>1$ and any $Q\in\mathcal{S}(a_2,\sigma)$; by Lemma~\ref{lemma:universal_acceptance_cutoff}, equivalently any admissible $Q\in[0,\bar Q(a_2,\sigma)]$. Then $N(R;Q,a_2,\sigma)\geq 0$ for all $R\in(0,1)$. By Step~1, $N(R;Q,a_1,\sigma)>N(R;Q,a_2,\sigma)\geq 0$ for every $R\in(0,1)$. At the auxiliary boundary points $R=0$ and $R=1$, the values $N(0;Q,a,\sigma)\ =\ q_b(q_g+Q)$ and $N(1;Q,a,\sigma)\ =\ q_g(q_b+Q)$ (computed in Step~6 of the proof of Lemma~\ref{lemma:universal_acceptance_cutoff}) are independent of $a$ and strictly positive, so $N(\,\cdot\,;Q,a_1,\sigma)\geq 0$ on $[0,1]$. Hence $Q\in\mathcal{S}(a_1,\sigma)$, which gives $\bar Q(a_2,\sigma)\leq\bar Q(a_1,\sigma)$.

\medskip
\textbf{Step 3: } \textit{Suppose} $a_2>a_1>1$. \textit{Then} $\bar Q(a_2,\sigma) < \bar Q(a_1,\sigma)$.

\textbf{ Proof of Step 3:} 
Suppose, toward a contradiction, that $\bar Q(a_1,\sigma)=\bar Q(a_2,\sigma)=:\bar Q$. By Lemma~\ref{lemma:universal_acceptance_cutoff}(ii), $\bar Q<1/2$, and by the relative closedness of $\mathcal{S}(a_2,\sigma)$ in the admissible domain $[0,1/2)$ established in Step~7 of the previous proof, $\bar Q\in\mathcal{S}(a_2,\sigma)$. Hence $N(R;\bar Q,a_2,\sigma)\geq 0$ for all $R\in[0,1]$. By Step~1, for every $R\in(0,1)$, $N(R;\bar Q,a_1,\sigma)>N(R;\bar Q,a_2,\sigma)\geq 0$. At the auxiliary boundary points, $ N(0;\bar Q,a_1,\sigma)\ =\ q_b(q_g+\bar Q)$ and $N(1;\bar Q,a_1,\sigma)\ =\ q_g(q_b+\bar Q)$, both of which are strictly positive since $\bar Q<1/2$ implies $q_g,q_b>0$. Therefore, $m\ :=\ \min_{R\in[0,1]}N(R;\bar Q,a_1,\sigma)\ >\ 0.$ Choose any $\delta>0$ with $\bar Q+\delta<1/2$. The function $N$ is uniformly continuous on the compact set $[0,1]\times[\bar Q,\bar Q+\delta]$ (with $a_1$ and $\sigma$ fixed), so there exists $\varepsilon\in(0,\delta]$ such that
\begin{align*}
    \bigl|N(R;Q,a_1,\sigma)-N(R;\bar Q,a_1,\sigma)\bigr|\ <\ m\qquad\text{for all }R\in[0,1]\text{ and all }Q\in[\bar Q,\bar Q+\varepsilon].
\end{align*}
For such $Q$, $N(R;Q,a_1,\sigma)>0$ for all $R\in[0,1]$, so receivers are universally obedient at $(Q,a_1,\sigma)$. This contradicts the definition of $\bar Q(a_1,\sigma)$ as $\sup\mathcal{S}(a_1,\sigma)$. We conclude $\bar Q(a_2,\sigma)<\bar Q(a_1,\sigma)$.

\medskip
\textbf{Step 4:} $\bar Q(a,\sigma)$ \textit{is strictly increasing in }$a$ \textit{on} $(0,1)$.

\textbf{Proof of Step 4:}
By Lemma~\ref{lemma:universal_acceptance_cutoff}(iii), for $a<1$ universal acceptance at $(Q,a,\sigma)$ is equivalent to $\tilde N(R;Q,a,\sigma)\geq 0$ for all $R\in(0,1)$, where
\begin{align*}
    \tilde N(R;Q,a,\sigma)\ :=\ q_gq_b-Q^2+Q(q_b+Q)(1-R^a)+Q(q_g+Q)\bigl(1-(1-R)^a\bigr).
\end{align*}
For fixed $R\in(0,1)$, only the terms $R^a$ and $(1-R)^a$ depend on $a$, so
\begin{align*}
    \frac{\partial \tilde N(R;Q,a,\sigma)}{\partial a}\ =\ -Q(q_b+Q)R^a\ln R-Q(q_g+Q)(1-R)^a\ln(1-R)\ >\ 0,
\end{align*}
since $\ln R<0$ and $\ln(1-R)<0$. Hence, if $0<a_1<a_2<1$ and $Q\in\mathcal{S}(a_1,\sigma)$, then $\tilde N(R;Q,a_1,\sigma)\geq 0$ for all $R\in(0,1)$, and Step~1's strict inequality yields $\tilde N(R;Q,a_2,\sigma)>\tilde N(R;Q,a_1,\sigma)\geq 0$. The boundary values $\tilde N(0;Q,a,\sigma)\ =\ q_b(q_g+Q)$ and $\tilde N(1;Q,a,\sigma)\ =\ q_g(q_b+Q)$ are independent of $a$ and strictly positive, so $\tilde N(\,\cdot\,;Q,a_2,\sigma)\geq 0$ on $[0,1]$, giving $Q\in\mathcal{S}(a_2,\sigma)$. Thus $\bar Q(a_1,\sigma)\leq\bar Q(a_2,\sigma)$. The strict-inequality argument from Step~3 carries over verbatim with $\tilde N$ in place of $N$ (using $\partial\tilde N/\partial a>0$ instead of $\partial N/\partial a<0$, and the relative closedness of $\mathcal{S}(a_1,\sigma)$). From this, Step 4 follows.
\end{proof}

Thus the permissible prevalence of controversial products is largest at the symmetric benchmark $a=1$ and shrinks as the population becomes more polarized in either direction.


    \end{document}